\documentclass[conference]{IEEEtran}
\usepackage{graphicx}
\usepackage[figuresright]{rotating}
\usepackage{epstopdf}
\usepackage{amsmath}
\usepackage{subfigure}
\usepackage{algorithm}
\usepackage{algorithmic}
\usepackage{booktabs}
\usepackage{float}
\usepackage{rotating}
\usepackage{array}
\usepackage{cite} 
\graphicspath{{Pictures/}}
\newcommand\pN{\mathcal{N}}

\usepackage{amsmath,amssymb,graphicx,adjustbox,multirow,hhline}

\usepackage[utf8]{inputenc}
\usepackage[english]{babel}
\usepackage{amsthm}
\newtheorem{theorem}{Theorem}[section]

\newtheorem{lemma}[theorem]{Lemma}

\long\def\comment#1{}

\newcommand{\beq}{\begin{equation*}}
\newcommand{\eeq}{\end{equation*}}

\newfont{\bbb}{msbm10 scaled 700}

\newfont{\bb}{msbm10 scaled 1100}

\newcommand{\nv}{{\bf n}}

\newcommand{\vv}{{\bf v}}
\newcommand{\xv}{{\bf x}}
\newcommand{\yv}{{\bf y}}

\newcommand{\thetav}{\hbox{\boldmath$\theta$}}

\newcommand{\Phim}{\hbox{\boldmath$\Phi$}}

\newcommand{\Psim}{\hbox{\boldmath$\Psi$}}

\newcommand{\Thetam}{\hbox{\boldmath$\Theta$}}

\renewcommand{\arg}{{\hbox{arg}}}

\begin{document}
\title{Layer-Adaptive Communication and Collaborative Transformed-Domain Representations for Performance Optimization in WSNs}

{\author{\IEEEauthorblockN{Muzammil Behzad$^{1}$, Manal Abdullah$^{2}$, Muhammad Talal Hassan$^{1}$, Yao Ge$^3$, Mahmood Ashraf Khan$^{1}$}
		\IEEEauthorblockA{$^1$COMSATS Institute of Information Technology (CIIT), Islamabad 44000, Pakistan\\
			$^2$King Abdulaziz University (KAU), Jeddah 21589, Kingdom of Saudi Arabia\\
			$^3$The Chinese University of Hong Kong (CUHK), Shatin 999077, Hong Kong\\
			Email: \{muzammil, talal\}@vcomsats.edu.pk, maaabdullah@kau.edu.sa, yge@ee.cuhk.edu.hk, mahmoodashraf@comsats.edu.pk
			}}}

\maketitle

\begin{abstract}
In this paper, we combat the problem of performance optimization in wireless sensor networks. Specifically, a novel framework is proposed to handle two major research issues. Firstly, we optimize the utilization of resources available to various nodes at hand. This is achieved via proposed optimal network clustering enriched with layer-adaptive 3-tier communication mechanism to diminish energy holes. We also introduce a mathematical coverage model that helps us minimize the number of coverage holes. Secondly, we present a novel approach to recover the corrupted version of the data received over noisy wireless channels. A robust sparse-domain based recovery method equipped with specially developed averaging filter is used to take care of the unwanted noisy components added to the data samples. Our proposed framework provides a handy routing protocol that enjoys improved computation complexity and elongated network lifetime as demonstrated with the help of extensive simulation results.
\end{abstract}


\IEEEpeerreviewmaketitle

\section{Introduction}
Latest accelerations toward technological advancements in wireless communication based Internet-of-Things (IoT) offer great deal of support to the research community dedicated to develop extremely convenient and elegant systems in terms of design, computational cost, and practical implementation. Such efforts have allowed the electronics industry to assemble wireless devices with tiny structure, economical value, and having the ability to efficiently use the available power resources. Toward this end, wireless sensor networks are gaining importance in the continual progress of information and communication technologies \cite{7110295}. However, since majority of the wireless devices are constrained by the available resources, the power consumption and communication overhead are therefore significant research areas for design and development of wireless networks toward efficient management in IoT.

Wireless sensor networks (WSNs) are made up of small, portable and energy-restricted sensor nodes deployed in an observation venue. These nodes carry the baggage of transmitting vital information using wireless radio links. Such information can have the form of multi-dimensional signals and are of critical importance in many world-wide applications. The development of such networks demand extensive planning strategy along with superior tactical approaches for~its working capabilities. This effective development motivates~the existence of many real-time application scenarios such as security monitoring \cite{6604028}, environmental control \cite{7887698}, under-water quality measurement \cite{7808887}, battlefield surveillance \cite{7528266}, target detection \cite{7997348}, remote environmental monitoring \cite{Yick20082292}, medical and healthcare \cite{6883890}, an amazingly increasing interest~toward application like refineries, petrochemicals, underwater development facilities, and oil and gas platforms \cite{5474813}, and many~more.

The first step in establishing a WSN is the initialization and distribution of sensor nodes around the observation field. Many researchers have advocated normal distribution of the statically deployed nodes as the optimal distribution (e.g., see \cite{Bharucha:2008:DPL:1461796.1461800, 7368292}). The deployment of nodes in the network field area is then followed by transmission of the required data. Since these nodes are left unattended, with limited resources at hand, efficient utilization of the available resources become a key operation factor to form a vigorous and standalone network.

To tackle the aforementioned problem, many protocols have exploited and voted for clustering of the network area, a pioneer contribution by W. B. Heinzelman \cite{heinzelman2000application}, as one of the appreciated methods toward efficient resource utilization. Fundamentally, clustering aims to effectively divide the network field into multiple smaller observation versions thereby making resource management a comparatively convenient task \cite{7016049, 7920983, behzad2017distributed}. However, this requires free and fair election of cluster heads (CHs) in each cluster. These CHs are responsible for data fusion, i.e., they receive data from their respective clusters' normal nodes, and transmit them to the base station (BS), also known as sink, for necessary actions.

Traditionally, researchers propose routing protocols for WSNs by focusing mainly on performance improvements via effective selection criterion during CH election, the choice of single-hop or multi-hop communication between nodes, and whether the clustering scheme should be static or dynamic. Even though an optimal combination of the above factors yield interesting results, which is a blessing in disguise, however, impressive results can be achieved by looking into more complex parts of the problem instead of merely scratching the tip of the iceberg. Furthermore, another much concerned and barely discussed side of the problem, while designing these routing protocols, is tackling the degradation due to environment. A common form of the inevitable degradation affecting the data sent over radio links is that of additive white Gaussian noise (AWGN). Several protocols and algorithms designed in this domain discard the aforesaid important issue and just focus on minimizing the energy consumption by assuming that the data received by nodes have not experienced any noise addition due to the environment.

\subsection{Recoveries via Sparse Representations}
 Contrary to the traditional Nyquist-Shannon sampling theorem, where a certain minimum number of samples is required in order to perfectly recover an arbitrary band-limited signal, i.e., one must sample at least double the signal bandwidth, compressive sensing (CS) has emerged out as a new framework for data acquisition and sensor design in an extremely competent way. The basic idea is that if the data signal is sparse in a known basis, a perfect recovery of the signal can be achieved leading to a significant reduction in the number of measurements that need to be stored \cite{qaisar2013compressive}.
 
 In this regard, a wide investigation has been performed which showed promising number of real-time CS applications for so many areas, such as imaging, radar, speech recognition, data acquisition, wireless communication and sensor networks \cite{qaisar2013compressive}. Its applications in communication are mainly on channel estimation, interference removal, Massive MIMO, and enable leveraging tools to design efficient estimation algorithms for wireless channels \cite{5454399,alkhateeb2014channel}. Since CS is a significantly promising approach for simultaneous sensing, data gathering and compression, this allows a large reduction in resources usability, computational costs, and thus prolongs the network lifetime for both small and large-scale WSNs.

According to CS, we consider the following linear measurement model to recover an unknown vector $\vv$ from an under-determined system:
\begin{equation}
\xv = \Phim \vv = \Phim \Psim \thetav = \Thetam \thetav,
\end{equation}
where $\xv \in {\mathbb{C}^M}$, $\vv = \Psim \thetav \in {\mathbb{C}^N}$ are observed signal and unknown vector, $\thetav\in {\mathbb{C}^N}$ is unknown sparse signal which, for example, a node will collect, representing projection coefficients of $\vv$ on $\Psim$, $\Thetam = \Phim \Psim$ is an $M \times N$ reconstruction matrix ($M < N$). The measurement matrix $\Thetam$, which satisfies the Restricted Isometric Property (RIP) \cite{donoho2006compressed}, is designed such that the dominant information of $\thetav$ can be captured into $\xv$.

Reconstruction algorithms in CS exploit the fact that many signals are genetically sparse, therefore they proceed to minimize ${\ell_0},{\ell_1}$ or ${\ell_2}$ norm over the solution space. Among them, $\ell_1$-norm is the most accepted approach due to its tendency to successfully recover the sparse estimate $\hat {\thetav}$ of $\thetav$ as follows:
\begin{align}\label{opt.CS}
\hat {\thetav} = \mathop {\arg \min }\limits_{\thetav} {\left\| {\thetav} \right\|_{{\ell_1}}}, \quad \text{subject to} \quad {\xv} \approx \Phim \Psim \thetav.
\end{align}
\begin{figure}[t]
	\centering
	\includegraphics[width=0.6\linewidth]{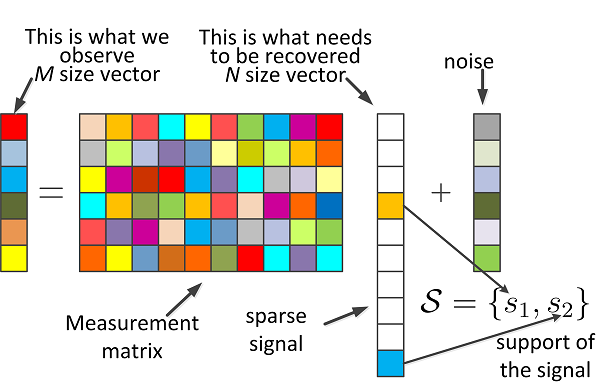}
	\caption{Sparse Model}
	\label{fig:fig1}
\end{figure}

In this regard, a number of methods are available to solve the optimization problem in (\ref{opt.CS}), such as Matching Pursuit (MP) \cite{tropp2004greed}, Basis Pursuit (BP) \cite{candes2006robust} and Linear Program (LP) \cite{boyd2004convex}. Despite this, however, the inevitable presence of noise in wireless channels is always a challenging task to combat. Consequently, the system is modeled as
\begin{equation}
	\yv = \xv + \nv = \Thetam \thetav + \nv,
\end{equation}
where $\yv \in {\mathbb{C}^M}$ is the noisy version of the clean signal $\xv \in {\mathbb{C}^M}$ which is corrupted by the noise vector $\nv \in {\mathbb{C}^M}$ with i.i.d. zero mean Gaussian entries having variance $\sigma_{\nv}^2$ i.e., \boldmath{$\nv$}(.) $\sim \pN(0, \sigma_{\nv}^2\textbf{1})$. A depiction of this model is shown in Fig.  \ref{fig:fig1}.

\subsection{Contribution}
In this paper, we introduce a novel framework to tackle two major concerns in WSNs; 1) optimization of the system performance via efficient energy utilization, and 2) combating the unavoidable presence of Gaussian noise, in the application data, added as a result of multiple communications among the deployed nodes via wireless channel. In our framework, we propose a fast and low-cost sparse representations based collaborative system enriched with layer-adaptive 3-tier communication mechanism. This is supported by an effective CHs election method and mathematically convenient coverage model guaranteeing minimization of energy and coverage holes. A computationally desired implementation of our proposed framework is an added benefit that makes it a preferable choice for real-world applications.

For performance optimization of the network, we deploy a 3-tier communication architecture, equipped with adaptively-elected CHs, among all the reactively sensing nodes where we make sure to restrict the nodes sending data signals over longer distances. To tackle the problem of AWGN, we transmit the data in its spatial domain form and later compute its sparse estimates at the receiver end for separating out the unwanted noisy components. For a much better denoising, we let the nodes situated at a single hop to mutually negotiate with each other for better collaboration. The denoising of the data is further refined by a specially designed averaging filter. Our proposed protocol lends itself the following salient features:
\begin{itemize}
	\item The implementation of a mathematically efficient coverage model along with an adaptive CHs election method help avoiding coverage holes to a greater extent.
	\item Our proposed layer-adaptive 3-tier communication system greatly reduce energy holes.
	\item To compute denoised data signals, we compute support-independent sparse estimates which relieves us from finding distribution of the sparse representations first, hence, giving it a support-agnostic nature.
	\item Prior collaboration enjoyed by the nodes for communication yields an effectively significant energy minimization.
	\item The use of a fast sparse recovery technique allows a desired computational complexity of our algorithm.
\end{itemize}

Rest of the paper is structured as follows: Section \ref{Related_Work} presents an overview of the related work done in the literature. We describe our proposed framework in Section \ref{Proposed_Framework}, while the results from various simulations are presented and discussed in Section \ref{Results_Discussions}. Finally, Section \ref{Conclusion} concludes the paper.

\section{Literature Review}
\label{Related_Work}
\subsection{Related Work}
Presently, the most trending researchers of WSNs are fundamentally concentrating on the technologically-enriched tools for performance optimization of network structure, the routing efficiency in communication networks, optimization of network structure, as a result of which the lifetime of these WSNs is possible to increase. This possibility provides a roadway for scientists working in this research domain to propose low-cost, energy-efficient and optimized algorithms~\cite{8109448}. 

This include a trend-setter work by W. B. Heinzelman, et al., proposing a multi-hop energy efficient communication protocol for WSNs, namely LEACH \cite{926982}. The objective was to reduce energy dissipation by introducing randomly elected CHs resulting in, however, unbalanced CHs distribution. Nevertheless, the partition of network area into different regions via clustering yielded a significant increase in system lifetime. As a principal competitor, a new reactive protocol, named as TEEN, was proposed by authors of \cite{925197} for event-driven applications. This protocol, even though constrained to temperature based scenarios only, proposed threshold aware transmissions thereby outperforming LEACH in terms of network lifespan.

In comparison with the aforementioned homogeneous WSNs, the authors of \cite{smaragdakis2004sep} and \cite{QING20062230} proposed SEP and DEEC, respectively, introducing heterogeneous versions of the WSNs by allowing a specific set of nodes, defined as advanced nodes, to carry higher initial energy than other normal nodes. SEP used energy based weighted election to appoint CHs in a two-level heterogeneous network ultimately improving network stability. As a stronger contestant to SEP, DEEC deployed multi-level heterogeneity and improvised CHs election measure to attain extended lifespan of the network than SEP.

A. Khadivi, et al., proposed a fault tolerant power aware protocol with static clustering (FTPASC) for WSNs in \cite{1696382}. The network was partitioned into static clusters and energy load was distributed evenly over high-power nodes, resulting in minimization of power consumption, and increased network lifetime. Another static clustering based sparsity-aware energy efficient clustering (SEEC) protocol is proposed in \cite{SWQ28921546.OW12N}. This protocol used sparsity and density search algorithms to classify sparse and dense regions. A mobile sink is then exploited, specifically in sparse areas, to enhance network lifetime.

As opposed to static clustering, authors of \cite{5136647} presented centralized dynamic clustering (CDC) environment for WSNs. In this protocol, the clusters and number of nodes associated with each cluster remains fixed, and a new CH is chosen in each round of communication between clusters and BS. CDC showed better results than LEACH in terms of communication overhead and latency. In a similar fashion, G.~S.~Tomar, et al., proposed an adaptive dynamic clustering protocol for WSNs in \cite{4809826}, which creates a dynamic system that can change topology architecture as per traffic patterns. Mutual negotiation scheme is used between nodes of different energy levels to form energy efficient clusters. Periodic selection of CHs is done based on different characteristics of nodes. Another work proposed to use the cooperative and dynamic clustering to achieve energy efficiency \cite{7822956}. This framework ensured even distribution of energy, and optimization of number of nodes used for event reporting thereby showing promising results in network lifetime.

D. Jia, et al., tackled the problem of unreasonable CHs selection in clustering algorithms \cite{7365420}. The authors considered dynamic CH selection methods as the best remedy to avoid overlapping coverage regions. Their experimental results showed increased network lifetime than LEACH and DEEC. Another energy efficient cluster based routing protocol, termed as density controlled divide-and-rule (DDR), is proposed in \cite{S2RT4387429.OW12N}. The authors tried to take care of the coverage and energy holes problem in clustering scenarios. They presented density controlled uniform distribution of nodes and optimum selection of CHs in each round to solve this issue. Similarly, a cluster based energy efficient routing protocol (CBER) is proposed in \cite{6531736}. This protocol elects the CHs on the basis of optimal CH distance and nodes' residual energy. CBER reported to outperform LEACH in terms of energy consumption of the network and its lifetime.

\subsection{Motivation}
Currently, most of the traditional protocols for WSNs burden the nodes with execution overhead yielding many unaddressed coverage and energy holes. This is because such protocols go for inadequate network configuration, uneven distribution of nodes, and using CHs criterion which cause undesirable results. Some other research works put efforts to make the network mobile, e.g., mobile sink or mobile nodes. However, the upshot of such a step makes farthest nodes in a network to wait longer for transmitting their data which could be critical based on the application scenarios.

Another unconsidered and important dimension of the problem is the lack of attention given to the fact that the communication among nodes via wireless channel is always corrupted by the presence of noise. Consequently, all the aforementioned factors motivated us to come up with our novel algorithm where we tackle the problems of both resource optimization for energy efficiency as well as denoising the degraded data.

\section{Proposed Framework Design}
\label{Proposed_Framework}
In this section, we provide the readers with detailed understanding of our proposed routing protocol. Here, we broadly discuss the widely accepted radio model for communication among nodes, which is then followed by a comprehensive explanation of our adopted network configuration and its operation details for energy efficiency and denoising of the~data.

\subsection{Wireless Communication Model}
For transmission and reception of required data among sensor nodes via wireless medium, we assume the simple and most commonly used first order radio communication model as given in Fig. \ref{fig:fig2}. In this figure, we present the energy consumed by a node while transmitting and receiving data. We show that a packet of data traveling over radio waves has to combat against degrading factors such as noise, multi-path fading, etc. Thus, we also take into account the $d^2$ losses that almost all chunks of data has to face. This is mathematically explained in terms of the following expressions:
\begin{equation}
E_{Tx}(k,d) =
\begin{cases}
k\times (E_{elec}+ \epsilon_{fs}\times d^2), & d < d_o\\
k\times (E_{elec}+ \epsilon_{mp}\times d^4), & d \geq d_o
\end{cases}
\end{equation}
\begin{equation}
\hspace{-5cm} E_{Rx}(k) = E_{elec}\times k,
\end{equation}
where $d_o$ is a reference distance, $k$ is the number of bits in packet, $d$ is the transmission distance which varies every time for each node, $E_{elec}$ is the energy used for data processing, $\epsilon_{fs}$ and $\epsilon_{mp}$ are channel dependent loss factors\footnote{This is worth noting that over larger distances, such loss factors demand a higher amount of energy yielding sudden death of the network. This is often missed by traditional protocols assuming lossless channel. Therefore, avoiding these power-hungry transmissions significantly optimize resources.}, $E_{Tx}$ is the energy used by a node for transmission, and $E_{Rx}$ is the energy used by a node for data reception. As shown, the $d^r$ losses may change from $\left. d^r \right \vert_{r=2}$ to $\left. d^r \right \vert_{r=4}$ forcing a higher value of $E_{Tx}$. A similar increase is then observed in the $E_{Rx}$ values to process a highly corrupted data when assuming a noisy environment, as in our case. The generally used energy dissipation values  for a radio channel are presented in Table~\ref{tab1}.
\begin{figure}[t]
	\centering
	\includegraphics[width=1\linewidth]{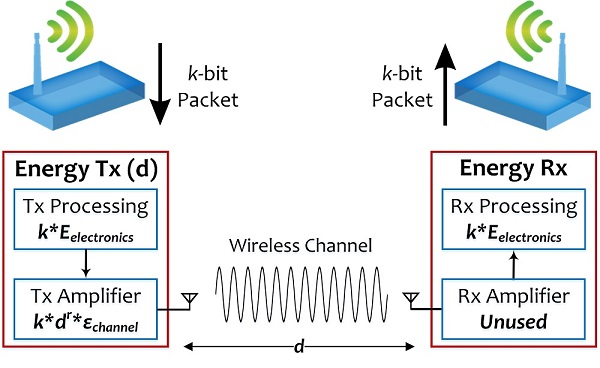}
	\caption{Radio Model}
	\label{fig:fig2}
\end{figure}
\begin{table}[b]
	\caption{\textbf{Energy Dissipation Measurements}}
	\centering
	\renewcommand{\arraystretch}{1}
	\begin{tabular}{|l|l|}
		\hline \hline
		\textbf{Dissipation Source} & \textbf{Amount Absorbed} \\ [0.5ex]
		\hline \hline
		$E_{elec}$ of Rx and Tx & 50 nJ/bit \\
		\hline
		Aggregation Energy  & 5 nJ/bit/signal\\
		\hline
		Tx Amplifier $\epsilon_{fs} \text{for} \left. d^r \right \vert_{r=2}$ & 10 pJ/bit/4m$^2$ \\
		\hline
		Tx Amplifier $\epsilon_{mp} \text{for} \left. d^r \right \vert_{r=4}$ & 0.0013 pJ/bit/m$^4$ \\
		\hline\hline
	\end{tabular}
	\label{tab1}
\end{table}

\subsection{Network Configuration}
\begin{figure*}[t]
	\centering
	\includegraphics[width=1\linewidth]{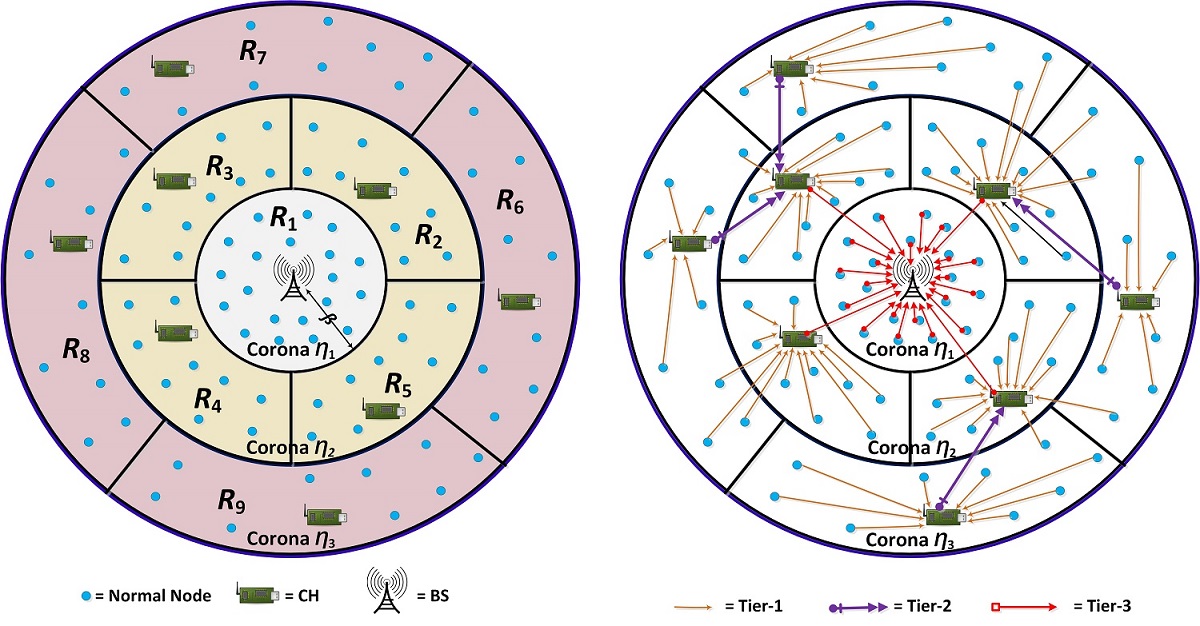}
	\caption{(a) Network Configuration Model, (b) 3-Tier Communication Architecture}
	\label{fig:fig3}
\end{figure*}
For the configuration model, we use a network consisting of $L$ number of nodes deployed randomly. Unlike the traditional models, we adopt a spherically-oriented field, and propose to use an optimized version of area division via adaptive clustering. For a much better understanding, see the network model shown in Fig. \ref{fig:fig3}(a). Here, for the sake of simplicity and understanding, we use the field area $A = \pi150^2\text{ m}^2$, i.e., the diameter $D = 300$, with a total of $L = 100$ deployed nodes. To avoid formation of energy holes, and thus the death of network, we place the BS in the center of the network at coordinates $<i,j> = (0,0)$. This is followed by the clustering of network field into various coronas which are then further classified into different sensing-based reporting regions. The prior computation of number of coronas, represented by $\eta$, is a function of field area $A$, which itself is depending on $D$, and the number of nodes $L$. As a sound approximate, we propose $\eta = D/L$. Hence, in our case we use $\eta$ = 3 coronas, denoted accordingly by $\eta_1, \eta_2 \text{ and } \eta_3$.

Once the $\eta$ number of coronas are formed, the next step is to divide each corona into various sensing regions as shown in the figure. However, for a much better network performance, the distribution is such that each sensing region in the upper level corona $\eta_\alpha$ surrounds two sensing regions in lower level corona $\eta_{\alpha-1}$. This is shown in Fig. \ref{fig:fig3}(a), where for example region $R_7 \text{ in } \eta_3$ covers both $R_2 \text{ and } R_3 \text{ in } \eta_2$, hence, avoiding coverage holes by satisfying the following expression for a general network configuration\footnote{Here, we presented calculations for $A = \pi150^2\text{ m}^2 \text{ and } L = 100$ merely for the ease of understanding. However, for any other small or large scale network configuration, the computations can be done in a similar fashion using the proposed expressions.}:
\begin{equation}
\label{eq:area}
	A = \pi(D/2)^2 = \sum_{\alpha=1}^{\eta} A_{\eta_{\alpha}} = \sum_{\alpha} A_{R_{\alpha}}, \hspace{0.365cm}\alpha \in \mathbb{Z}^+,
\end{equation}
where $A_{\eta_{\alpha}}$ and $A_{R_{\alpha}}$ represented area of each corona and sensing region, respectively. This is worth noting that we do not divide the corona surrounding BS further, $\eta_1$ in our case, to avoid unneeded and poor use of available resources. Thus, we can safely write:
\begin{equation}
	A_{\eta_{1}} = A_{R_{1}} = \pi\beta^2, \text{ where }\beta \in \mathbb{R}^+.
\end{equation}

\subsection{Nodes Deployment and Layer-Controlled CHs Nomination}
As soon as the network is clustered out into various coronas and sensing regions, the next step is to distribute the nodes randomly over these regions. To optimize resources, a sensible decision is to deploy an equal percentage of nodes over different regions to ensure minimization of coverage holes, and elongation of network lifetime. Therefore, in this scenario, we propose to deploy 20$\%$ of the nodes in region $R_1$ and the rest 80$\%$ of the nodes to be distributed evenly over $R_{2,3,...,8,9}$ regions as shown in Fig. \ref{fig:fig3}(a). This nodes' deployment always depend upon the network field area and number of nodes. Hence, for any other network configuration, an adjusted percentage can be calculated to optimize communication among nodes, and to avoid formation of energy and coverage holes.

Following the deployment of nodes and prior network initialization, the election of CHs is carried out in all $R_{2,3,...,8,9}$ regions. Since the use of CHs in clustering techniques plays an important role to improve network lifespan, effective criterion for CHs election is equally necessary for further improving performance of the network. The most commonly used measures for electing CHs are residual energy and distance from BS. We propose a blend of both to increase the life of each node. Furthermore, we introduce a layering-based election of CHs. This means that the election will take place in lower level coronas $\eta_{\alpha}$ first, and will then move to high level coronas $\eta_{\alpha+1}$ for higher level CHs. The reason to adopt this is the effectiveness noted in CHs election. Thus, in each round, all the nodes are assessed based on their residual energies and top 5$\%$ of the nodes having highest residual energies in their respective regions are shortlisted. These shortlisted nodes then contest against each other where the node with smallest distance to the CHs of both associated regions in lower level corona is elected as CH. The nodes in $\eta_2$ are evaluated in a similar fashion based on the residual energy but having minimum distance with the center of their respective region.

\subsection{Layer-Adaptive 3-Tier Communication Mechanism}
For transfer of data among various nodes, we propose a layer-adaptive 3-tier architecture. Our communication mechanism is enriched with distance-optimized transmissions to avoid wastage of energy. The nodes use a multi-hop scheme instead of directly transmitting the data of interest to BS. In tier-1 phase, all the normal nodes gather data and send it to the nearest CH. This CH may not necessarily be the same region CH. Here, we allow nodes in $\eta_\alpha$ to transmit the data to CHs of even $\eta_{\alpha-1}$. This is the reason why we distributed the sensing regions in such a way that each region in upper level corona is bordered with two regions in lower level corona. However, the nodes of a region in $\eta_\alpha$ cannot transmit to another region on the same corona, i.e., it must either send data to its own CH in $\eta_\alpha$, or any other nearest CH in the two bordered regions on lower level corona $\eta_{\alpha-1}$ as shown in Fig. \ref{fig:fig3}(b).

In the next tier-2 phase of communication, the CHs of $\eta_\alpha$ aggregate their data and then send it to the CHs of $\eta_{\alpha-1}$. Note that even though the CH of $R_3$ is receiving data from CHs of both $R_7 \text{ and } R_8$, this is another blessing in disguise. This is because, as shown in the figure, the CHs of both $R_7 \text{ and } R_8$ have not received data from all the nodes in its region, since some nodes find another nearest CH, so these CHs are aggregating and then forwarding a comparatively smaller amount of load thereby not over burdening themselves. Also, the CHs change in each round based on the election criterion, so it ultimately saves energy. Finally in tier-3 phase, all the CHs in lower level coronas send their data to the BS, hence, completing the data transmission process.

\subsection{Coverage Model}
For reduction in coverage holes, we express the coverage scenario of nodes by a mathematical model. All the deployed sensor nodes are represented in set notation as $\kappa = {\mu_1, \mu_2, \mu_3, ..., \mu_L}$. The coverage model of one alive node $\mu_\alpha$ belonging to the set $\kappa$ can be expressed as a sphere centered at $<i_\alpha,j_\alpha>$ with radius $h_\alpha$. We let a random variable $\aleph_\alpha$ define an event when a data pixel $<a,b>$ is within the coverage range of any node $\mu_\alpha$. As a result, the equivalent of likelihood of the event $\aleph_\alpha$ to happen, as denoted by $P\{\aleph_\alpha\}$, is represented as $P_{cov}\{a,b,\mu_\alpha\}$. A decomposed version of the above is given as follows:
\begin{multline}
	P\{\aleph_\alpha\} = P_{cov}\{a,b,\mu_\alpha\}=\\
	\begin{cases}
	1, & (a-i_\alpha)^2 + (b-j_\alpha)^2 \le h_\alpha^2\\
	0, & \text{otherwise}
	\end{cases}
\end{multline}
where the equation translates that a data pixel $<a,b>$ is surrounded by the coverage range of any random node $\mu_\alpha$ if the distance between them is smaller than the threshold radius $h_\alpha$. However, since the event $\aleph_\alpha$ is stochastically independent from others, this means $\left. h_\alpha \text{ and } h_\gamma \text{ are not related} \implies \alpha, \gamma \in [1,L] \text{ and }\alpha \ne\gamma\right.$. This gives us the following conclusive equations:
\begin{equation}
	P\{\overline{\aleph_\alpha}\} = 1-P\{\aleph_\alpha\} = 1-P_{cov}\{a,b,\mu_\alpha\},
\end{equation}
\begin{multline}
P\{\aleph_\alpha \cup \aleph_\gamma \} = \\ 1-P\{\overline{\aleph_\alpha}\cap\overline{\aleph_\gamma}\}= 1-P\{\overline{\aleph_\alpha}\}.P\{\overline{\aleph_\gamma}\},
\end{multline}
where $P\{\overline{\aleph_\alpha}\}$ denotes the statistical complement of $P\{\aleph_\alpha\}$ which means that $\mu_\alpha$ failed to assist data pixel $<a,b>$. Importantly, this data pixel is given coverage if any of the nodes in the set is covering it otherwise a coverage hole would form. Hence, the following expressions denote the probability such that data pixels would be within the coverage range of at least one of the nodes in the set to minimize coverage holes:
\begin{multline}
P_{cov}\{a,b,\kappa\} = P\{\bigcup\limits_{\alpha=1}^{L}\aleph_\alpha\} = 1 - P\{\bigcap\limits_{\alpha=1}^{L}\overline{\aleph_\alpha}\},\\
						=  1 - \prod_{\alpha=1}^{L}(1-P_{cov}\{a,b,\mu_\alpha\}).
\end{multline}
For further facilitation, we present the coverage rate as fraction of area under coverage, denoted by $Q$, and the overall area of the observation field as follows:
\begin{equation}
	P_{cov}\{\kappa\}=\sum_{\alpha=1}^{L}\sum_{\gamma=1}^{L}\frac{P_{cov}\{a,b,Q\}}{A}
\end{equation}

\subsection{Data Denoising}
After taking care of the energy efficiency, a major problem is to retrieve the original data back. This is because the received data is generally degraded by AWGN so it is of no use unless denoised. For this purpose, we propose denoising of the data samples via Bayesian analysis based sparse recovery techniques. To do so, we take into account the data correlation of various adjacent nodes and use this as an important piece of information for collaboration among various nodes. 

We use three stages for CS based sparse recovery technique to denoise the data. In doing so, received data is converted to sparse domain first (e.g., wavelet transform for images data). This is followed by computing similar and correlated data by adjacent nodes, giving them weights based on the similarity extent. Using equivalent sparse representations of data samples, probability of active taps is computed giving us the location of undesired corrupted support locations. With the help of correlation information, an averaging based collaborative step is performed to remove the unwanted noisy components as shown via flowchart in Fig \ref{fig:fig4}. Finally, we apply a specially developed averaging filter to further smooth out the data. This filter fundamentally works on finding similar data samples and then averaging those samples to provide a clean estimate of the data. Using a CS based pre-determined dictionary, a reverse transform is applied to give back the denoised data in spatial-domain representation as $\hat{\xv} = \Thetam \hat{\thetav}$. 

This should be noted that the Gaussianity property of the noisy data received and then aggregated at the receiver (e.g., CHs or BS) should remain intact. This is because, even though our proposed Bayesian analysis based denoising algorithm is agnostic to support distribution of the sparse coefficients, it does need the data samples to be corrupted by Gaussian noise collectively. A concise version of this is provided in the following Lemma \ref{lemma:lemma1} to support the accuracy of our denoising algorithm.
\begin{lemma}
	The aggregated data samples received at either CHs or BS keep the Gaussianity property intact, hence we can denoise the cumulative version of the AWGN corrupted data.
	\label{lemma:lemma1}
\end{lemma}
\begin{proof}
	To show this, we consider two independent Gaussian random data samples $P \text{ and } Q$ sent by nodes $\left.\mu_\alpha \text{ and } \mu_\gamma, \text{ both } \in \kappa\right.$. For data aggregated by CH, we let $\left.Z=\rho P+\delta Q\right.$. Without loss of generality, let $\rho \text{ and }\delta$ be positive real numbers because for $\rho<0$, $P$ would be replaced by $-P$, and then we would write $\mid \rho\mid$ instead of $\rho$. The commutative probability function can be written as:
	\begin{multline}
		F_Z(z)=P\{Z\le z\}=P\{\rho P+\delta Q\le z\}\\
		=\int\int_{\rho P+\delta Q\le z}\varphi(p)\varphi(q)
		dpdq
	\end{multline}
	where $\varphi(.)$ represents the unit Gaussian density function. However, as the integrand $(2\pi)^{-1}\exp(-(p^2+q^2)/2)$ possesses circular symmetry, the numerical property of this integral is a function of length of the origin from $\left.\rho p+\delta q=z\right.$. Consequently using coordinates rotation, we can conclude
	\begin{multline}
		F_Z(z)=\int_{p=-\infty}^{\zeta}\int_{q=-\infty}^{q=\infty}\varphi(p)\varphi(q)dpdq=\Delta(\zeta)
	\end{multline}
	where $\zeta=\frac{z}{\sqrt{\rho^2+\delta^2}},\text{ and }\Delta(.)$ shows standard Gaussian CDF. Hence, the CDF of $Z\vert_{L=2}$ is a zero-mean Gaussian random variable having total variance equal to $\left.\rho^2+\delta^2\right.$.
\end{proof}

\begin{figure}[t]
	\centering
	\includegraphics[width=0.85\linewidth]{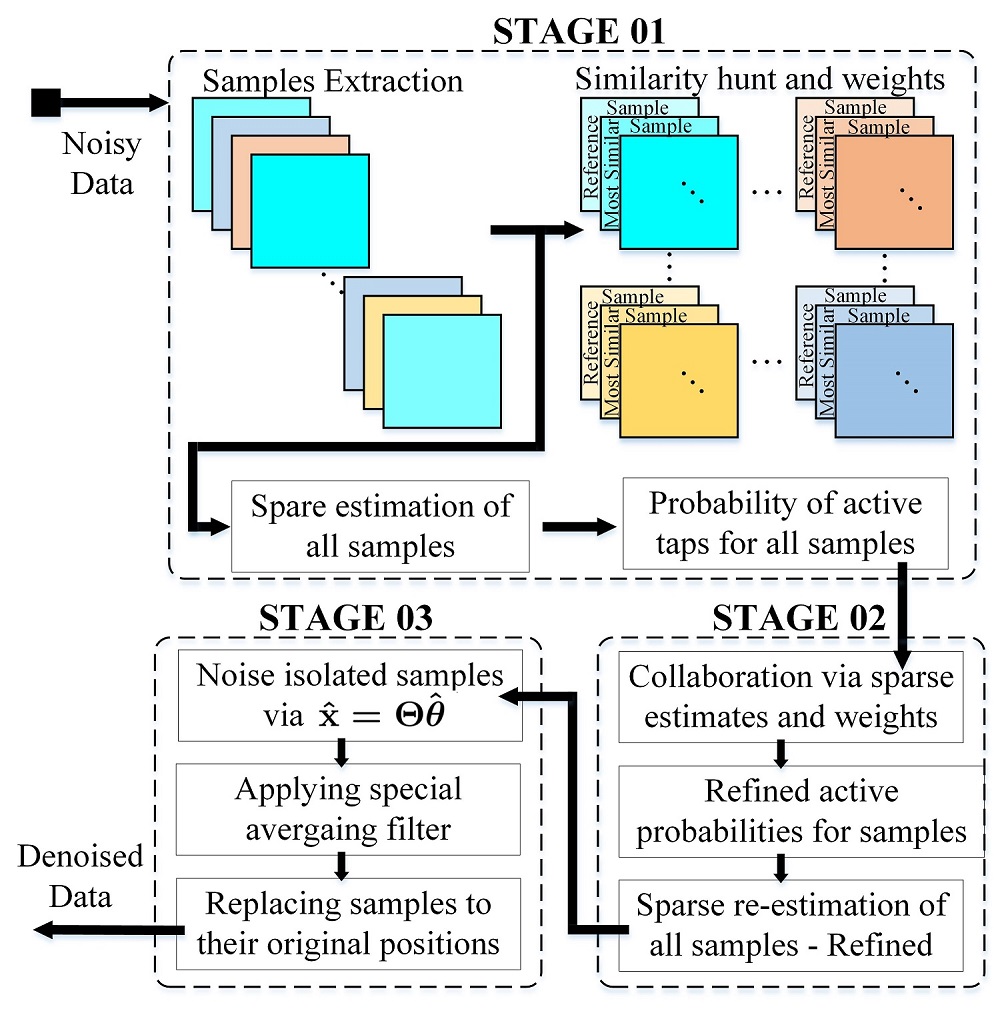}
	\caption{Flowchart of Data Denoising}
	\label{fig:fig4}
\end{figure}

\section{Results and Discussions}
\label{Results_Discussions}
In this section, we compare our proposed scheme with the state-of-the-art and traditional routing protocols such as LEACH \cite{926982}, TEEN \cite{925197}, SEP \cite{smaragdakis2004sep}, DEEC \cite{QING20062230} and DDR \cite{S2RT4387429.OW12N}. We use the values given in Table \ref{tab1}, and our experimentation is divided into two main scenarios: 1) efficient resource utilization, and 2) data denoising. The comparison is carried out over $\left.L = 100, 1000 \text{ and } 10000\right.$ nodes with following metrics: stability and instability period, network lifetime, energy consumption, computational complexity, peak signal-to-noise ratio (PSNR) and structural similarity (SSIM)~index.

A comparison of stability period for $L = 100$ is shown in Fig. \ref{fig:fig5}. This figure demonstrates the number of alive nodes over 8000 sensing rounds. It is evident from the figure that our proposed scheme significantly outperforms all the protocol and shows promising results. The first node die time of our approach is around 2900, while that of LEACH, TEEN, SEP, DEEC and DDR is around 800, 1900, 1600, 2000, and 1400, respectively. Similarly, Fig. \ref{fig:fig6} illustrates the all node die time (ADT) of these protocols for $L =$ 100. It can be clearly seen that the ADT of our method is $\sim$6390, $\sim$5290, $\sim$5490, $\sim$5190 and $\sim$4290 better than LEACH, TEEN, SEP, DEEC and DDR, respectively. We show that our scheme provides the best ADT and hence is the most suitable candidate for practical applications. This increase in ADT is associated with our proposed layer-adaptive 3-tier architecture equipped with effective CHs election mechanism.

\begin{figure}[t]
	\centering
	\includegraphics[width=1\linewidth]{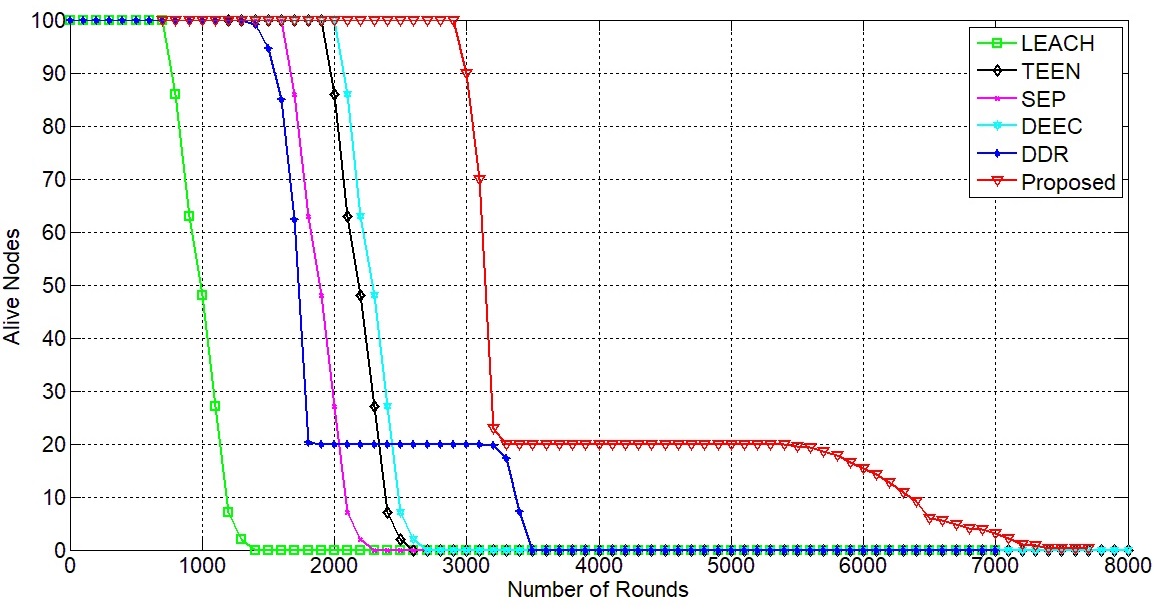}
	\caption{Stability Period}
	\label{fig:fig5}
\end{figure}
\begin{figure}[b]
	\centering
	\includegraphics[width=1\linewidth]{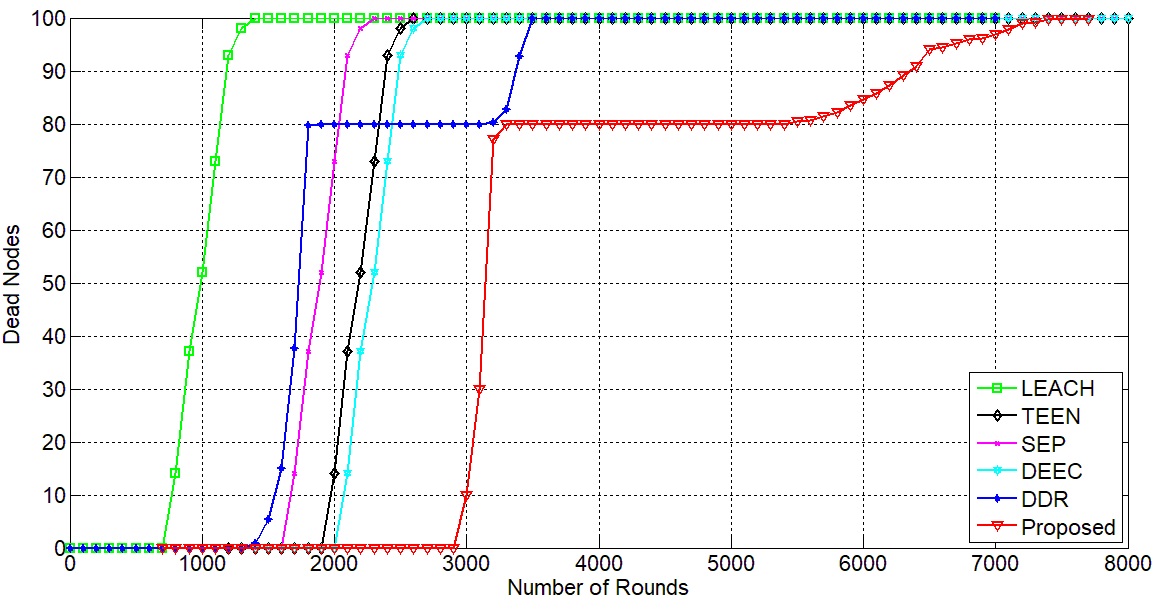}
	\caption{Instability Period}
	\label{fig:fig6}
\end{figure}

We provide a comparison of energy efficient resource utilization in Fig. \ref{fig:fig7}. Here, we show that all protocols start with same energy levels including both reactive and proactive techniques. However, based on the optimized communication method, our proposed scheme demonstrates outstanding results beating all the contestants. In Fig. \ref{fig:fig8}, we compare the network lifetime of our proposed method with LEACH and DDR for $\left.L = 100, 1000 \text{ and } 10000\right.$. It is validated that our protocol is equally competitive on large scale network scenarios outperforming each of the traditional methods.

The computational complexity of our approach is mainly dominated by the communication among nodes which fortunately yield a computationally convenient implementation of our method as compared with other protocols. A depiction of this is shown in Fig. \ref{fig:fig9}, where we compare the computational time consumed by the contestant methods using a $\left.\text{2.20 GHz Intel Core i7-3632QM}\right.$ machine for different number of nodes. This figure proves the robustness of our protocol by showing superior performance, hence, lending itself the most preferable choice for real-time applications.

\begin{table*}[h]
	\caption{Comparison of Denoising Image Data Samples in Terms of PSNR and SSIM}
	\centering
	\begin{tabular}{|c|c|c|c|c|c|c|c|c|c|c|}\hline\hline
		\multicolumn{2}{|c|}{Noise Level $\sigma_n$} & Cameraman & Lena & Barbara & House & Peppers & Living Room & Boat & Mandrill \\ \cline{1-10}
		\hline\hline
		\multirow{2}{*}{5}
		& Denoised & \textbf{37.94/0.85} & \textbf{36.97/0.94} & \textbf{36.38/0.98} & \textbf{39.24/0.71} & \textbf{36.60/0.96} & \textbf{36.37/0.96} & \textbf{36.24/0.94} & \textbf{35.42/0.95}  \\ \cline{2-10}
		& Noisy & 34.04/0.65 & 34.11/0.88 & 34.17/0.95 & 34.06/0.63 & 34.16/0.92 & 34.16/0.92 & 34.12/0.89 & 34.18/0.93 \\ \hline
		
		\multirow{2}{*}{10}
		& Denoised & \textbf{33.28/0.75} & \textbf{32.64/0.90} & \textbf{31.88/0.94} & \textbf{35.28/0.67} & \textbf{32.00/0.92} & \textbf{31.66/0.90} & \textbf{31.59/0.86} & \textbf{30.88/0.85} \\ \cline{2-10}
		& Noisy & 28.07/0.53 & 28.03/0.76 & 28.18/0.87 & 28.07/0.51 & 28.08/0.81 & 28.21/0.80 & 28.08/0.75 & 27.99/0.80 \\ \hline
		
		\multirow{2}{*}{15}
		& Denoised & \textbf{31.21/0.69} & \textbf{30.44/0.86} & \textbf{29.56/0.91} & \textbf{32.63/0.61} & \textbf{29.74/0.89} & \textbf{29.44/0.85} & \textbf{29.11/0.76} & \textbf{28.51/0.75}  \\ \cline{2-10}
		& Noisy & 24.56/0.45 & 24.63/0.66 & 24.59/0.77 & 24.57/0.44 & 24.72/0.72 & 24.62/0.68 & 24.59/0.63 & 24.56/0.66  \\ \hline
		
		\multirow{2}{*}{20}
		& Denoised & \textbf{29.23/0.63} & \textbf{28.77/0.82} & \textbf{27.93/0.88} & \textbf{31.33/0.58} & \textbf{28.11/0.85} & \textbf{27.76/0.79} & \textbf{27.48/0.69} & \textbf{27.09/0.67}  \\ \cline{2-10}
		& Noisy & 22.09/0.40 & 22.16/0.58 & 22.09/0.69 & 22.02/0.38 & 22.13/0.63 & 22.07/0.58 & 22.05/0.53 & 21.98/0.54  \\ \hline
		
		\multirow{2}{*}{25}
		& Denoised & \textbf{27.90/0.60} & \textbf{27.59/0.78} & \textbf{26.36/0.83} & \textbf{29.96/0.55} & \textbf{26.60/0.81} & \textbf{26.43/0.72} & \textbf{26.45/0.63} & \textbf{26.32/0.62}  \\ \cline{2-10}
		& Noisy & 20.22/0.36 & 20.13/0.50 & 20.12/0.60 & 20.16/0.33 & 20.19/0.56 & 20.01/0.49 & 20.37/0.46 & 20.07/0.45 \\ \hline
		
		\multirow{2}{*}{50}
		& Denoised & \textbf{24.19/0.45} & \textbf{23.97/0.62} & \textbf{22.75/0.68} & \textbf{25.80/0.45} & \textbf{23.09/0.68} & \textbf{23.43/0.54} & \textbf{23.42/0.45} & \textbf{24.17/0.47} \\ \cline{2-10}
		& Noisy & 14.13/0.21 & 14.11/0.28 & 14.10/0.33 & 14.03/0.19 & 14.17/0.33 & 14.19/0.25 & 14.17/0.24 & 14.16/0.20 \\ \hline
		
		\multirow{2}{*}{100}
		& Denoised & \textbf{20.67/0.25} & \textbf{20.50/0.48} & \textbf{20.11/0.50} & \textbf{22.14/0.27} & \textbf{19.62/0.49} & \textbf{20.96/0.33} & \textbf{20.64/0.26} & \textbf{21.30/0.28} \\ \cline{2-10}
		& Noisy & 08.18/0.10 & 08.13/0.11 & 08.21/0.13 & 08.09/0.07 & 08.16/0.13 & 08.10/0.09 & 08.10/0.09 & 08.18/0.06 \\ \hline\hline
	\end{tabular}
	\label{tab2}
\end{table*}
\begin{figure}[h]
	\centering
	\includegraphics[width=1\linewidth]{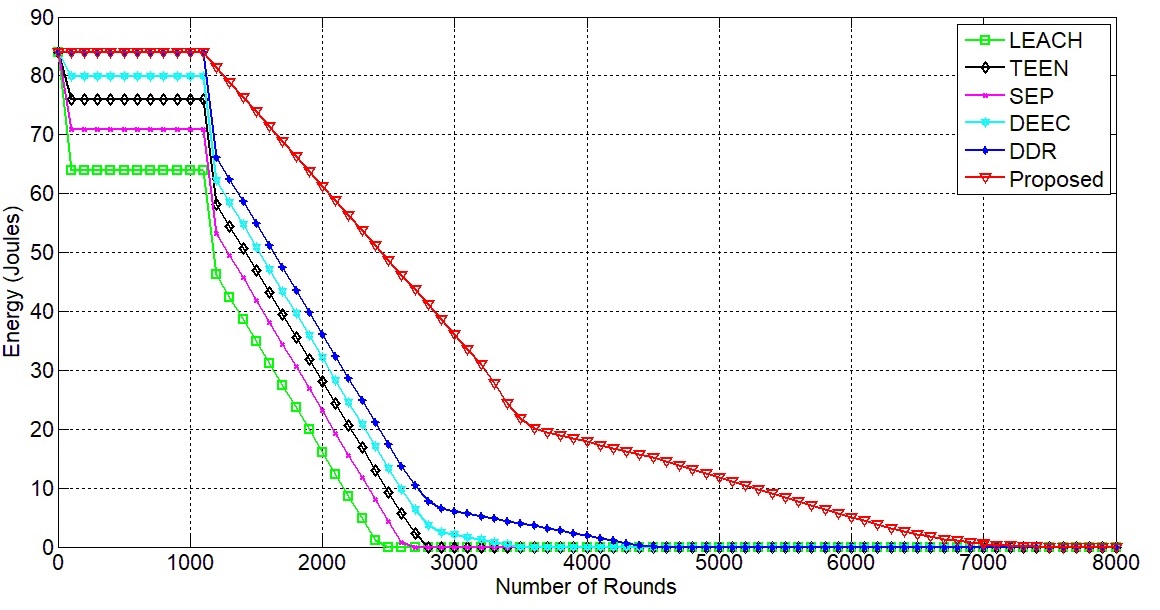}
	\caption{Energy Utilization Comparison}
	\label{fig:fig7}
\end{figure}
\begin{figure}[h]
	\centering
	\includegraphics[width=1\linewidth]{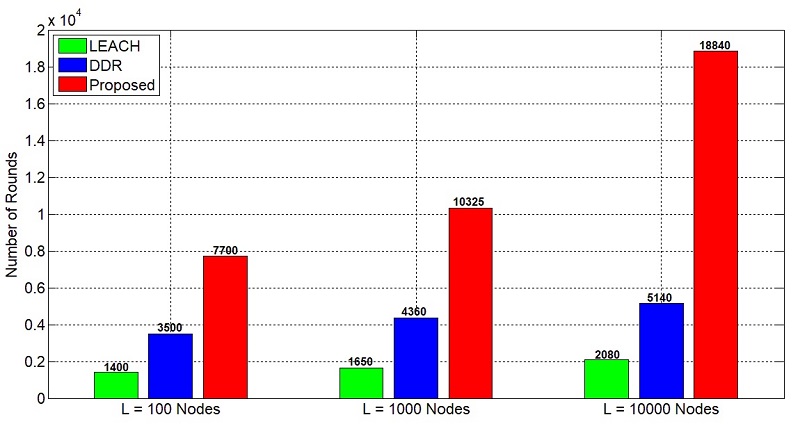}
	\caption{Network Lifetime}
	\label{fig:fig8}
\end{figure}
\begin{figure}[h]
	\centering
	\includegraphics[width=1\linewidth]{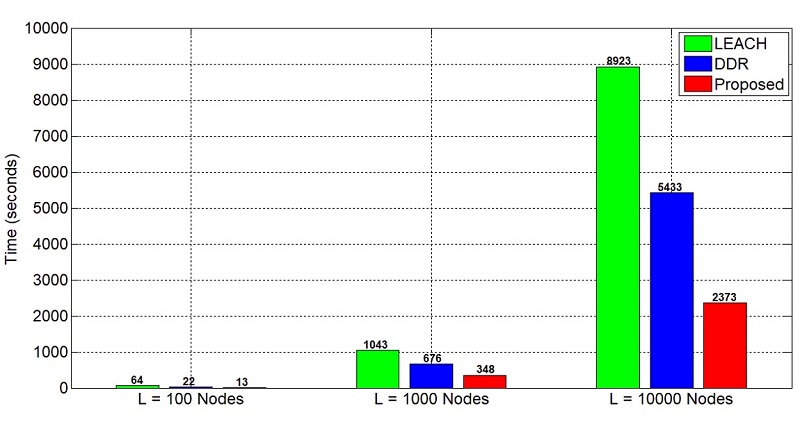}
	\caption{Computational Overload Comparison}
	\label{fig:fig9}
\end{figure}


Finally, the detailed denoising results of various standard test data images are presented in Table \ref{tab2}. We opt globally adopted PSNR and SSIM as evaluation metrics to prove that the denoising section of our proposed framework produces equally promising outcomes. The provided table summarizes denoising results of a number of images, as PSNR/SSIM, over a range of noise levels, i.e., $\left.\sigma_n = [5,10,15,20,25,50,100]\right.$. For experimentation, we transmitted various images among deployed nodes and showed that the resultant images received at the receiver suffers from Gaussian noise. The PSNR and SSIM values of the corresponding received noisy images are shown in the table. In comparison with our denoised images, we show that a significant amount of improvement is achieved in terms of the noise being removed, and the actual data is recovered to a greater extent. Consequently, these results confirm that our proposed framework is indeed an effective and robust model for real-time scenarios in WSNs which outperforms many traditionally proposed routing protocols.

\section{Conclusion}
\label{Conclusion}
We proposed a novel framework for energy efficient and data recoverable transmissions in WSNs. Equipped with layer-adaptive communication architecture, we managed to potentially minimize the energy holes. The presented coverage model helped us optimize the coverage area thereby reducing coverage holes. We also presented a sparse-domain based collaborative denoising scheme to combat the unwanted noisy components introduced in the data. Our proposed protocol afforded itself a computationally convenient complexity and outperformed traditional protocols by a significant margin.

%

\newpage
\begin{figure*}[t]
	\centering
	\includegraphics[width=1\linewidth]{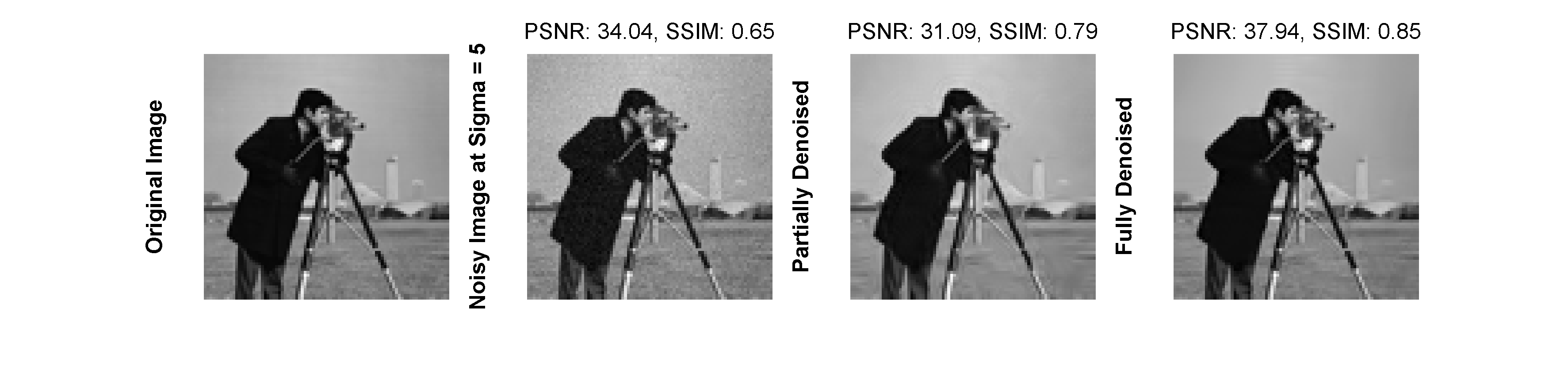}
	\includegraphics[width=1\linewidth]{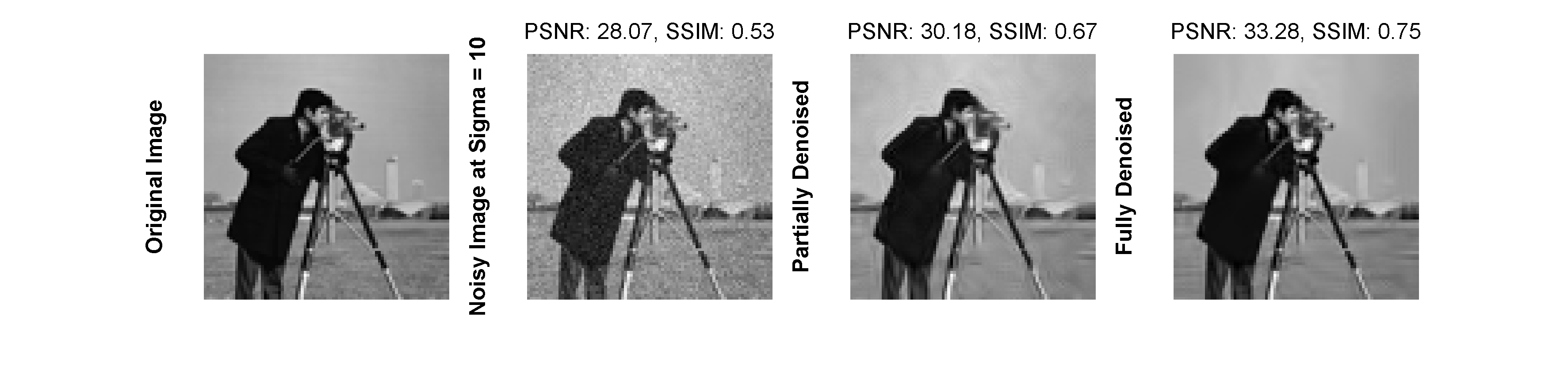}
	\includegraphics[width=1\linewidth]{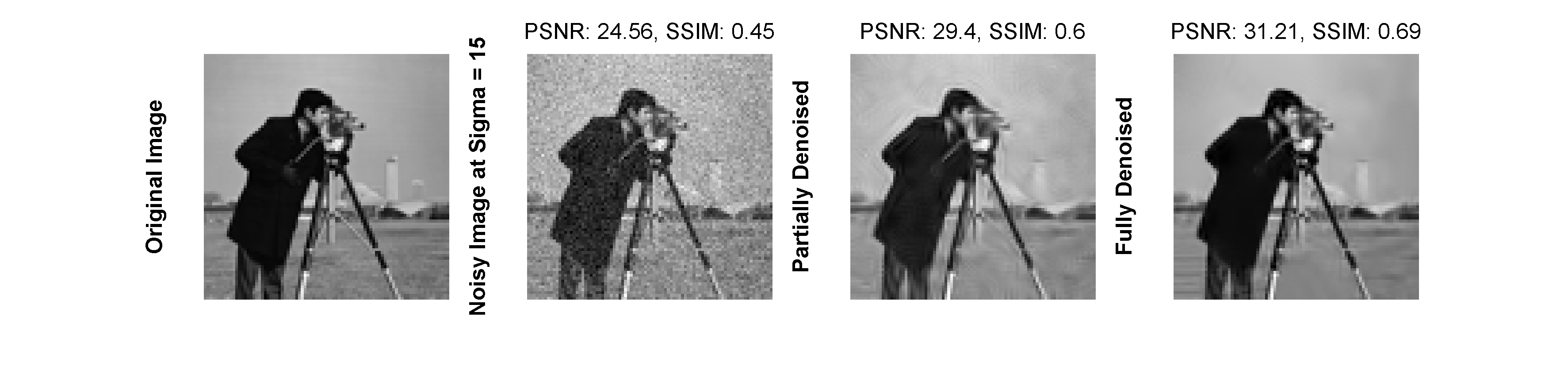}
	\includegraphics[width=1\linewidth]{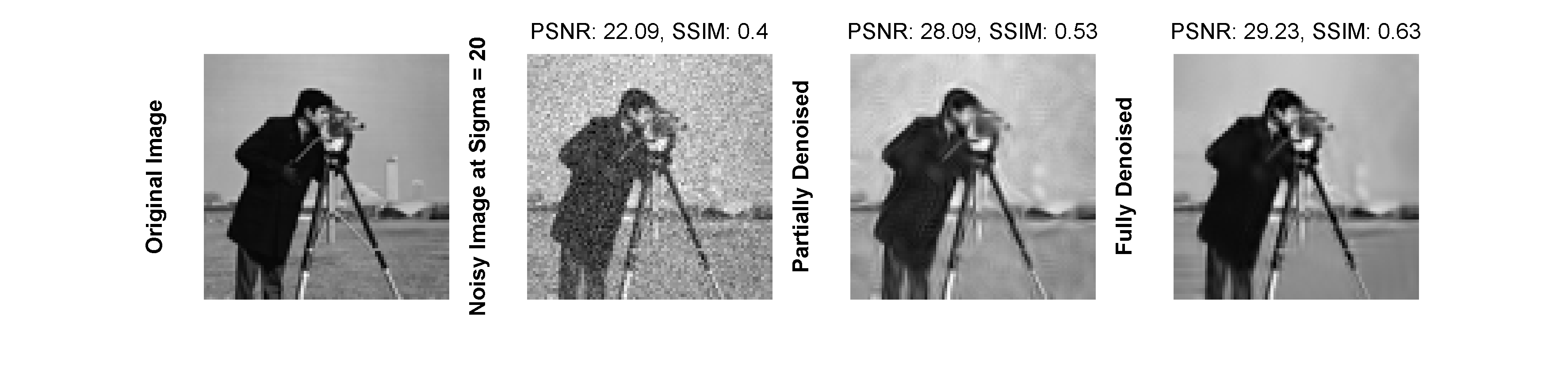}
	\includegraphics[width=1\linewidth]{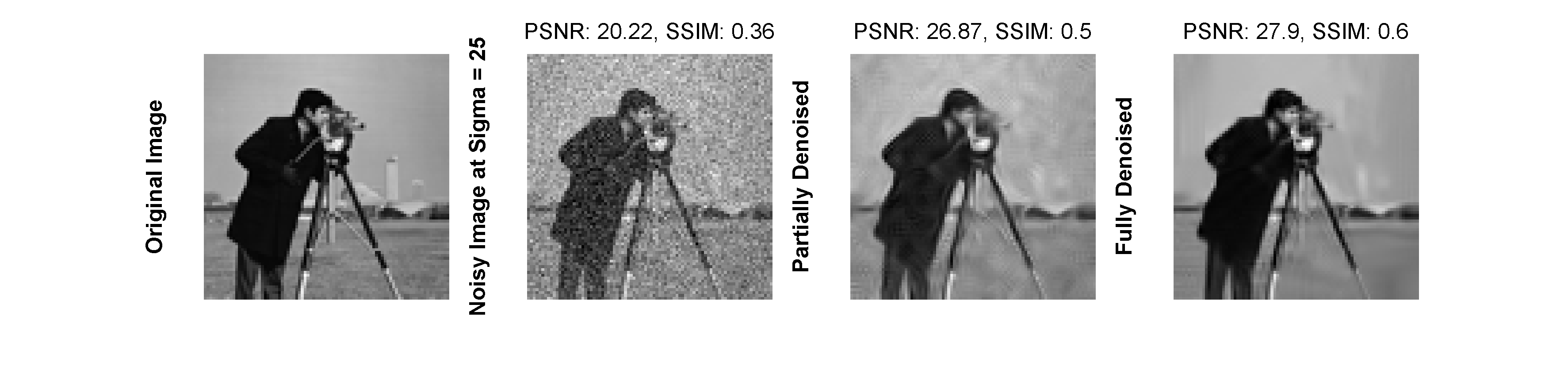}
\end{figure*}
\newpage
\begin{figure*}[t]
	\centering
	\includegraphics[width=1\linewidth]{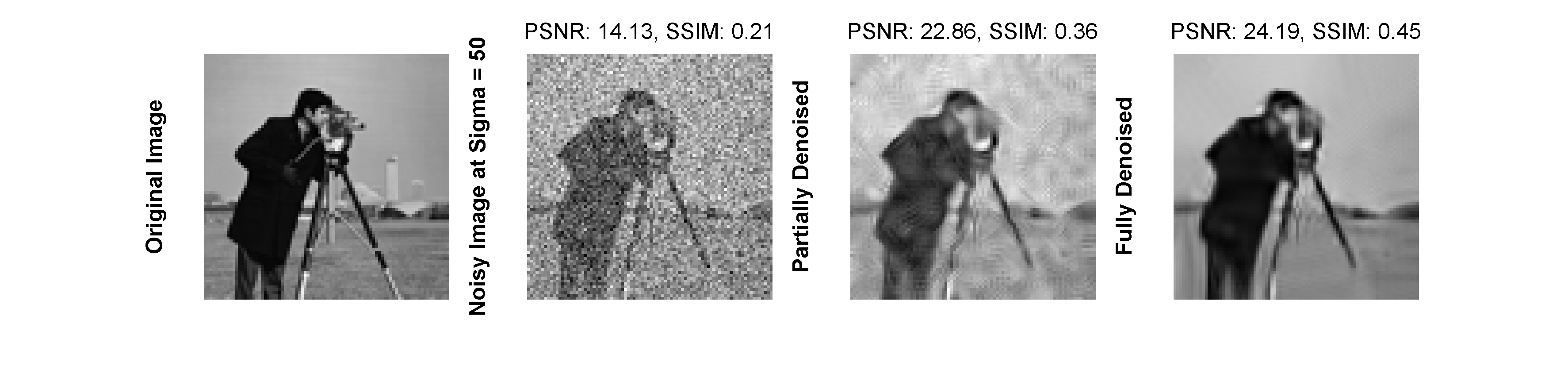}
	\includegraphics[width=1\linewidth]{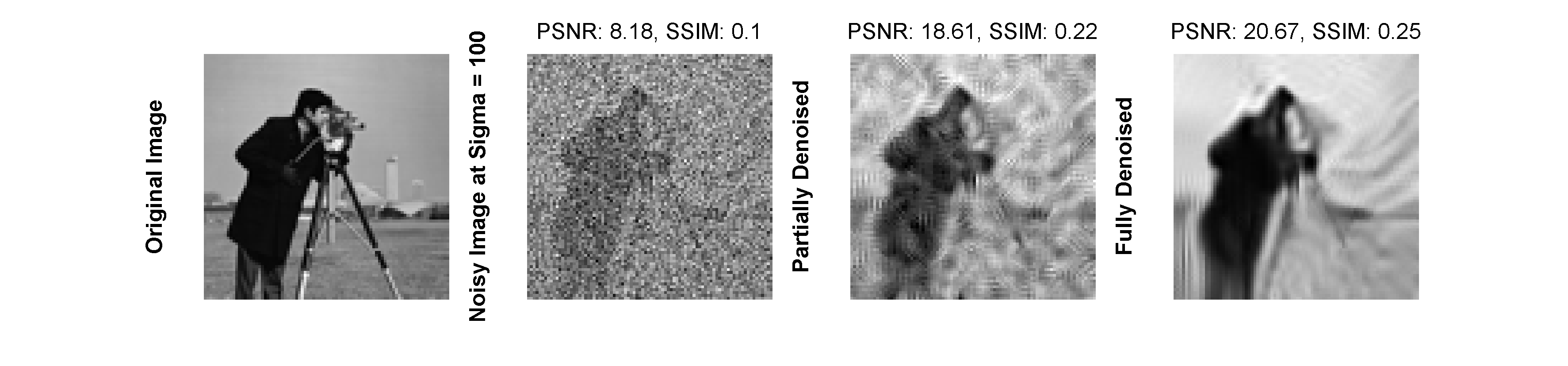}
	\includegraphics[width=1\linewidth]{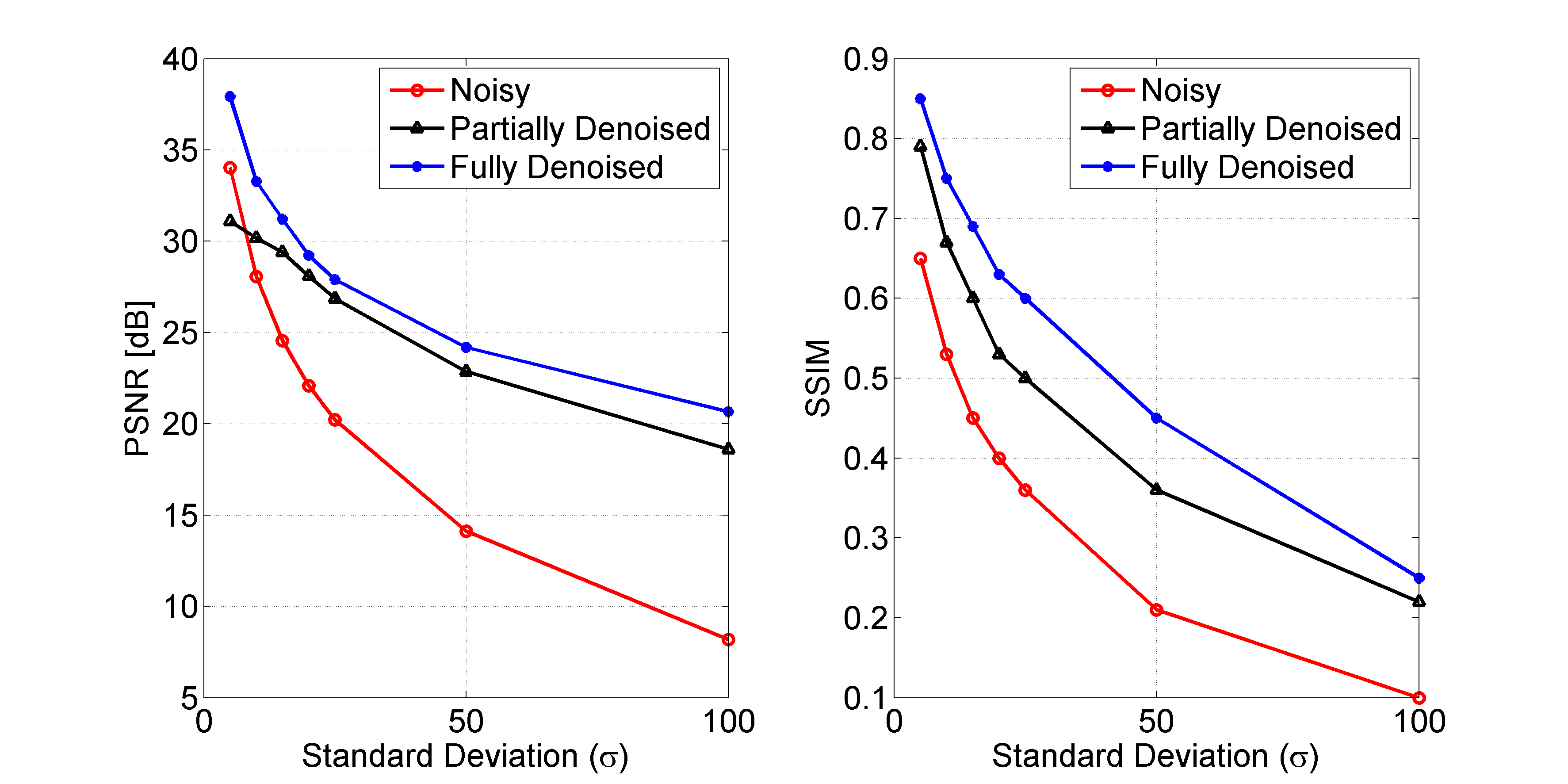}
	\caption{Denoising $256 \times256$ grayscale \textit{Cameraman} standard test data images over noise $\sigma =  [5,10,15,20,25,50,100]$ when received at a node $\mu_\alpha$. Each row represent an original image, a noisy image, a partially denoised, and a fully denoised image, respectively, corrupted by a specific level of additive white Gaussian noise (AWGN). The graphical results in the end show PSNR [dB] and SSIM results in the form of graphs.}
\end{figure*}

\newpage
\begin{figure*}[t]
	\centering
	\includegraphics[width=1\linewidth]{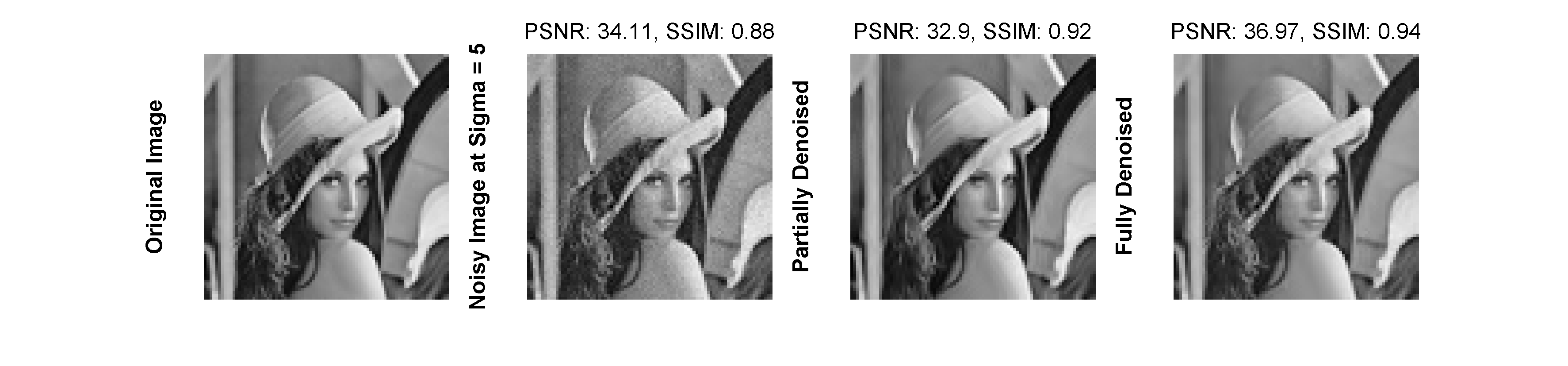}
	\includegraphics[width=1\linewidth]{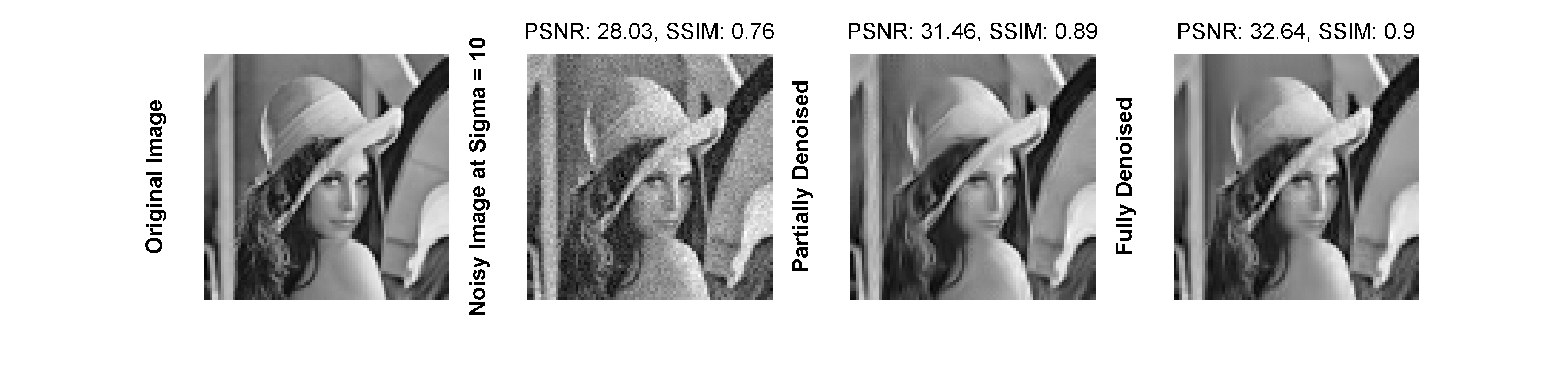}
	\includegraphics[width=1\linewidth]{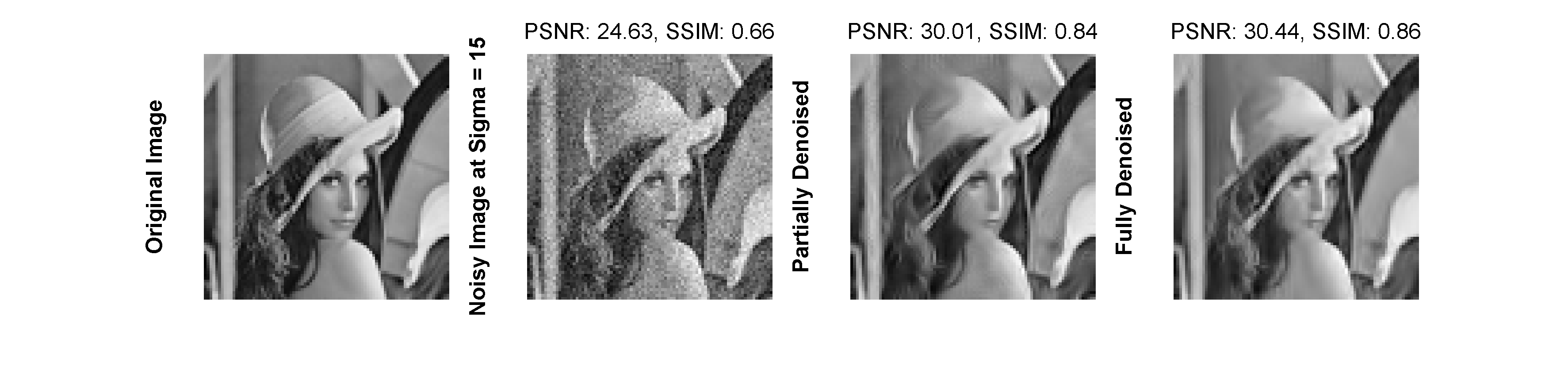}
	\includegraphics[width=1\linewidth]{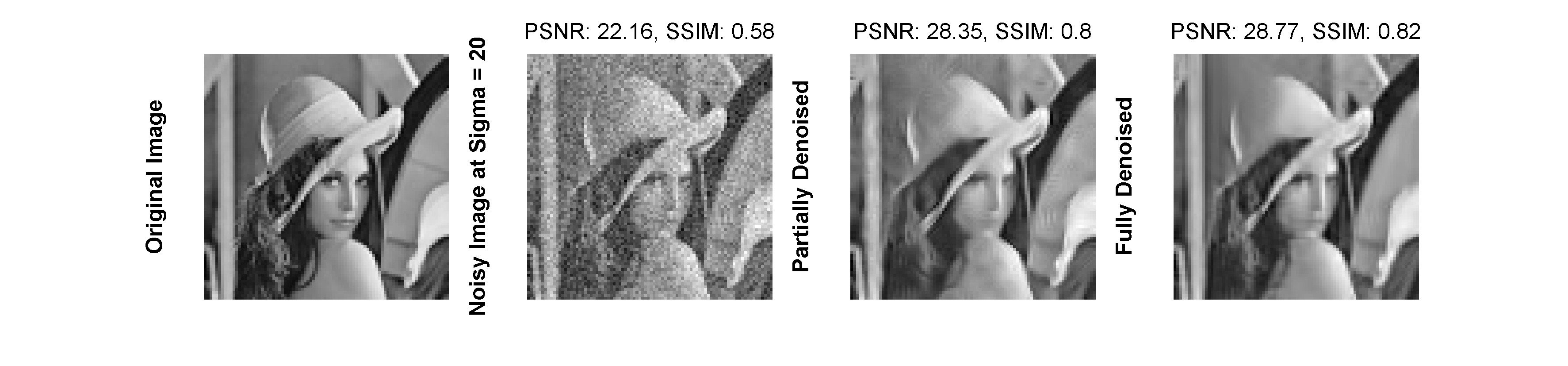}
	\includegraphics[width=1\linewidth]{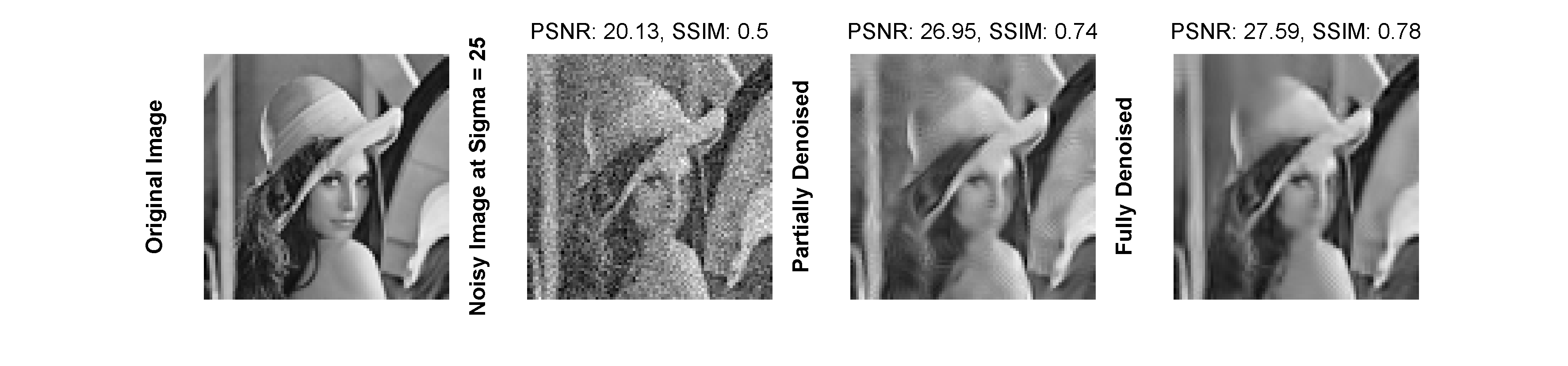}
\end{figure*}
\newpage
\begin{figure*}[t]
	\centering
	\includegraphics[width=1\linewidth]{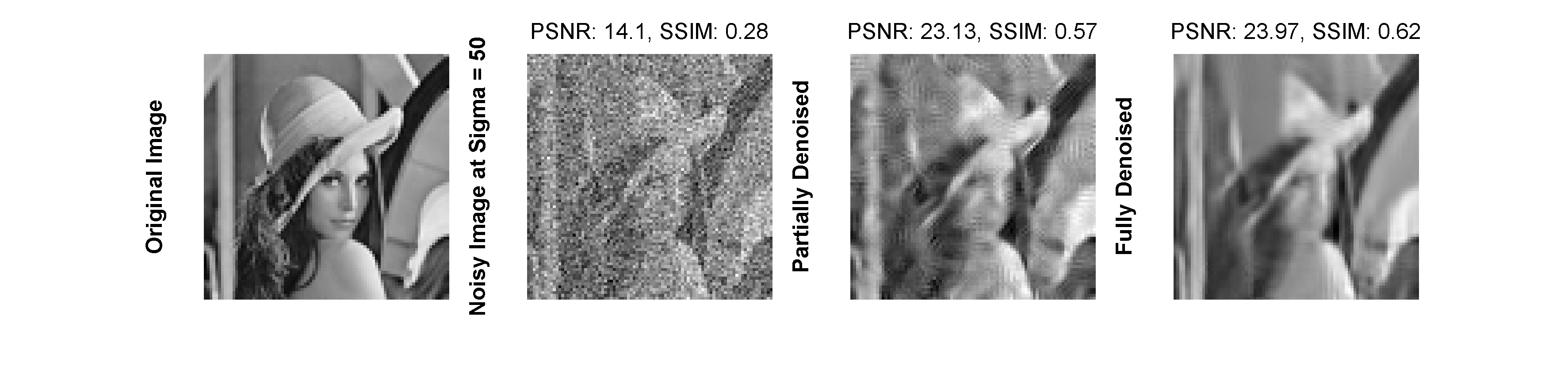}
	\includegraphics[width=1\linewidth]{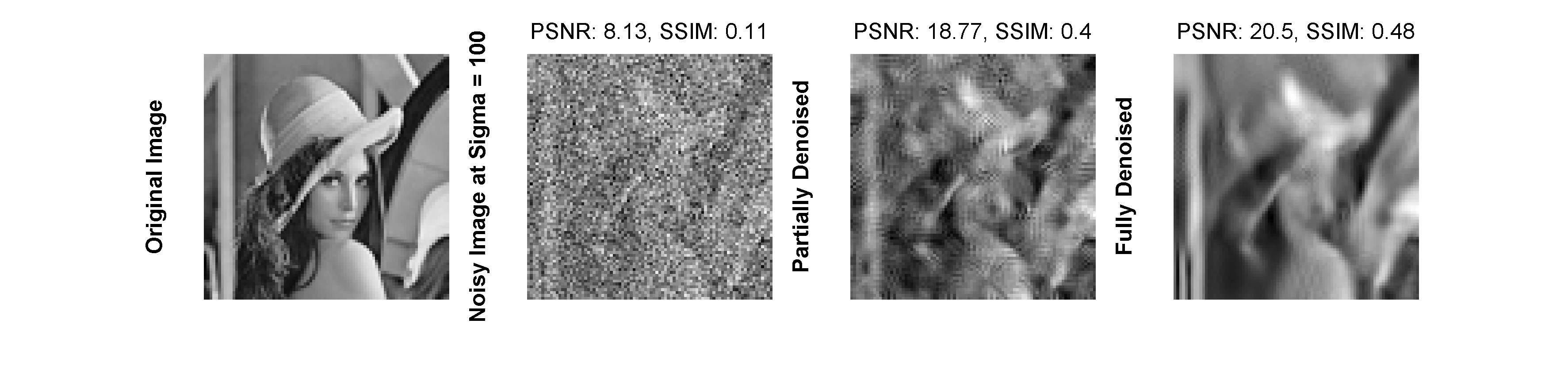}
	\includegraphics[width=1\linewidth]{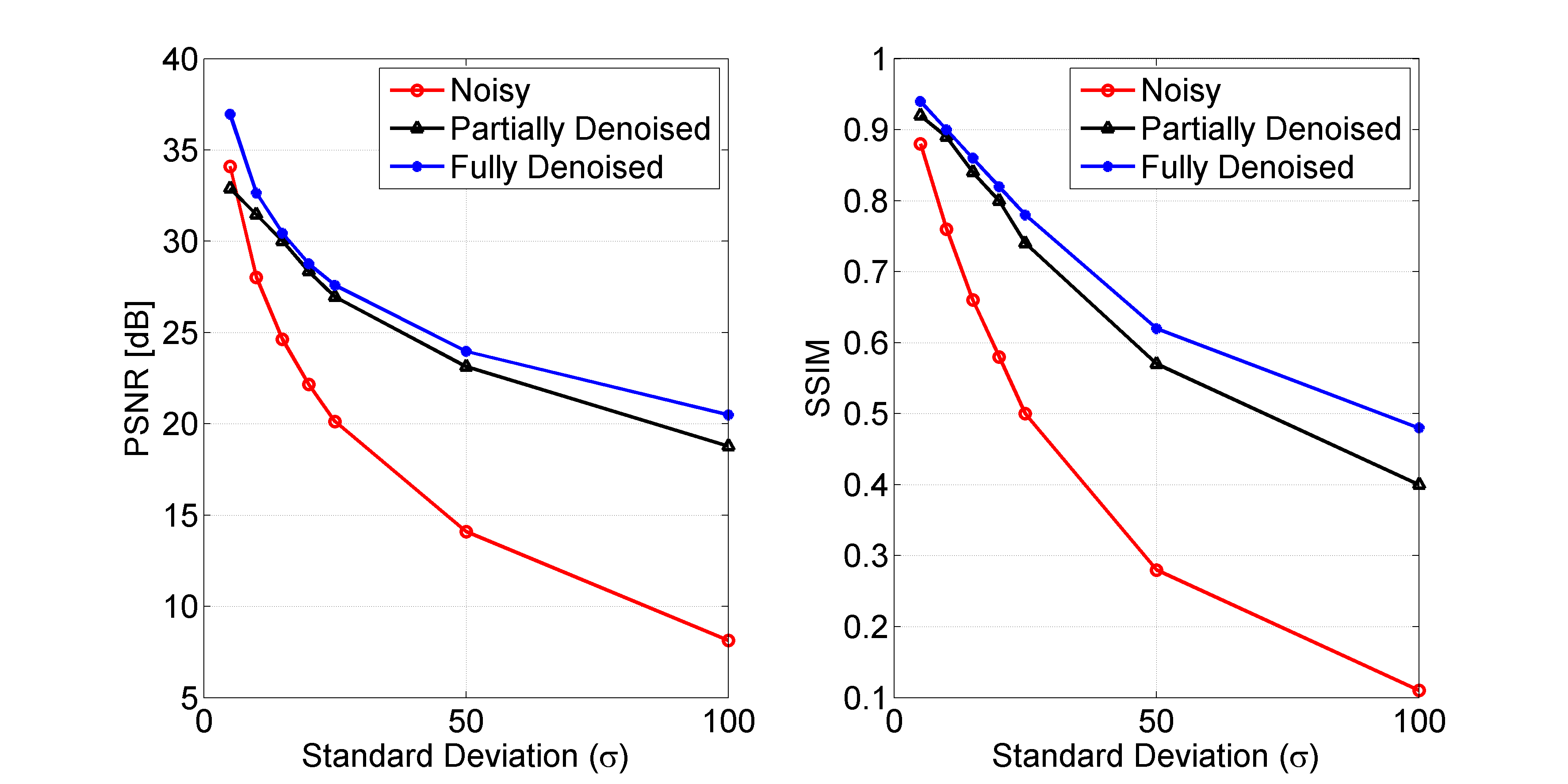}
	\caption{Denoising $256 \times256$ grayscale \textit{Lena} standard test data images over noise $\sigma =  [5,10,15,20,25,50,100]$ when received at a node $\mu_\alpha$. Each row represent an original image, a noisy image, a partially denoised, and a fully denoised image, respectively, corrupted by a specific level of additive white Gaussian noise (AWGN). The graphical results in the end show PSNR [dB] and SSIM results in the form of graphs.}
\end{figure*}

\newpage
\begin{figure*}[t]
	\centering
	\includegraphics[width=1\linewidth]{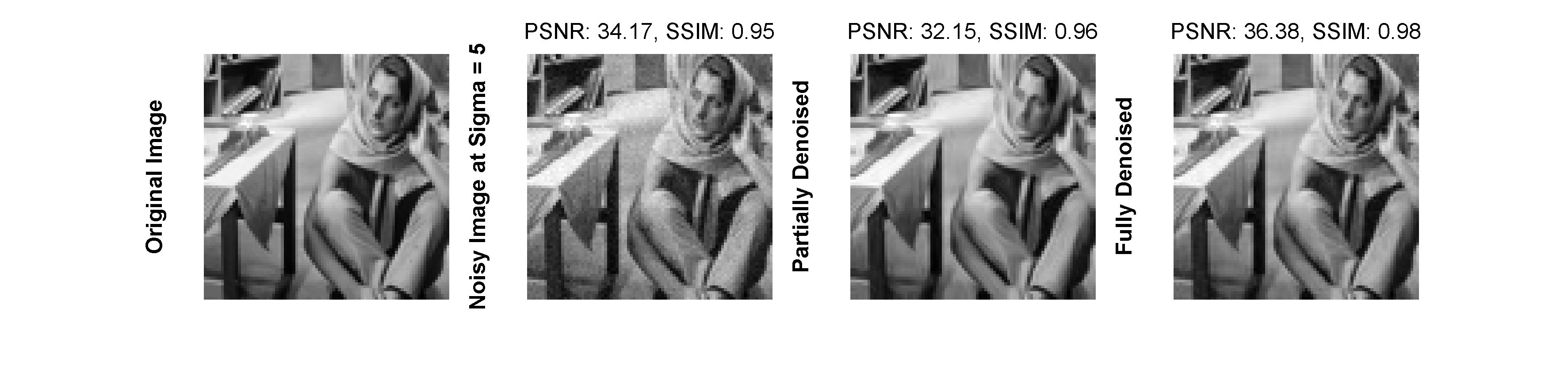}
	\includegraphics[width=1\linewidth]{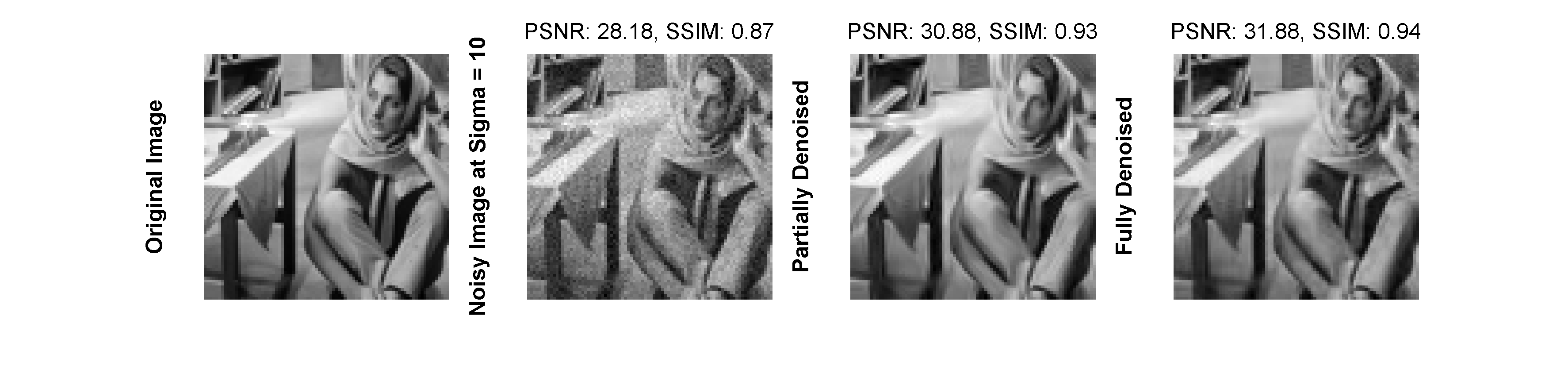}
	\includegraphics[width=1\linewidth]{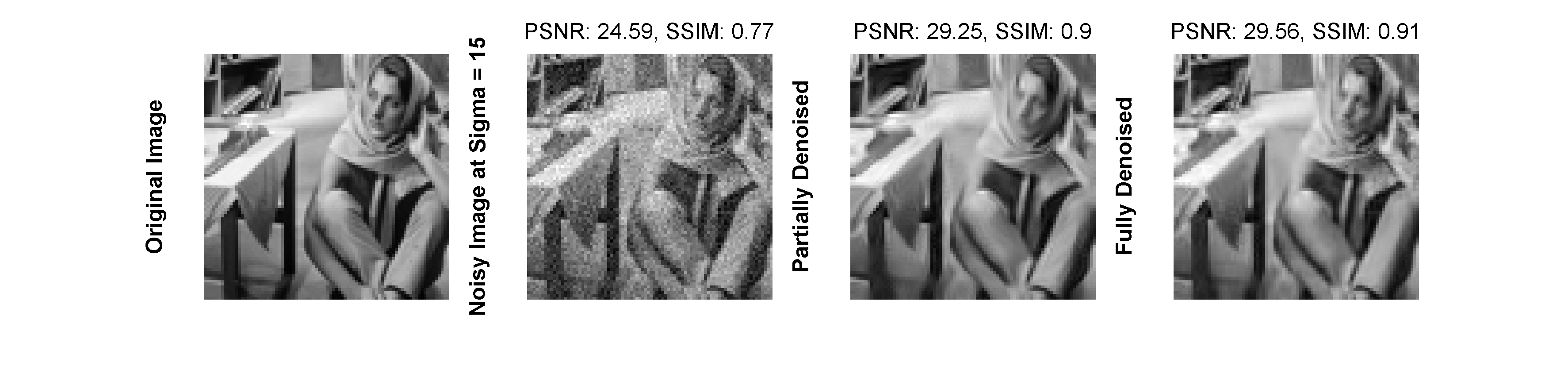}
	\includegraphics[width=1\linewidth]{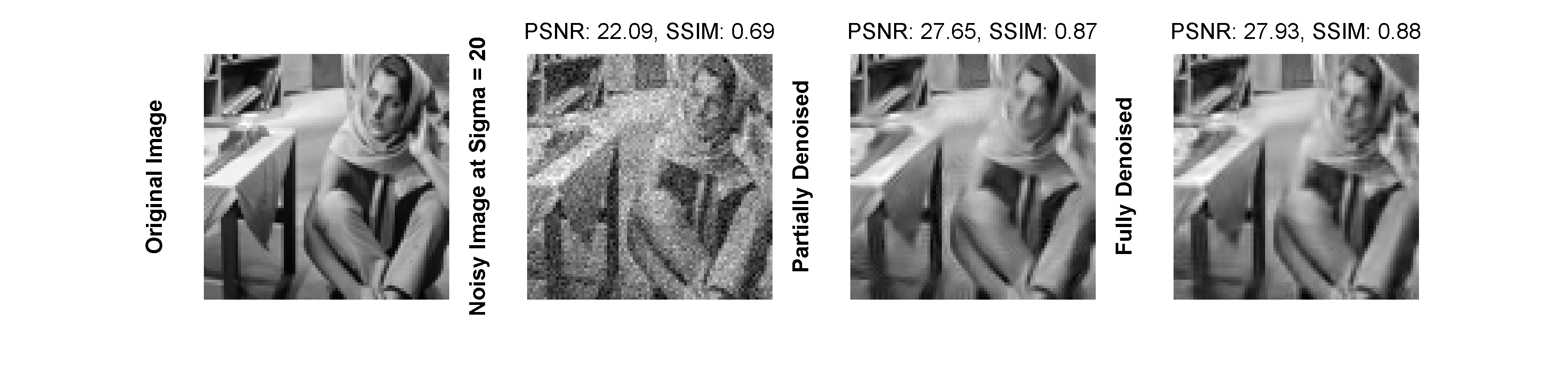}
	\includegraphics[width=1\linewidth]{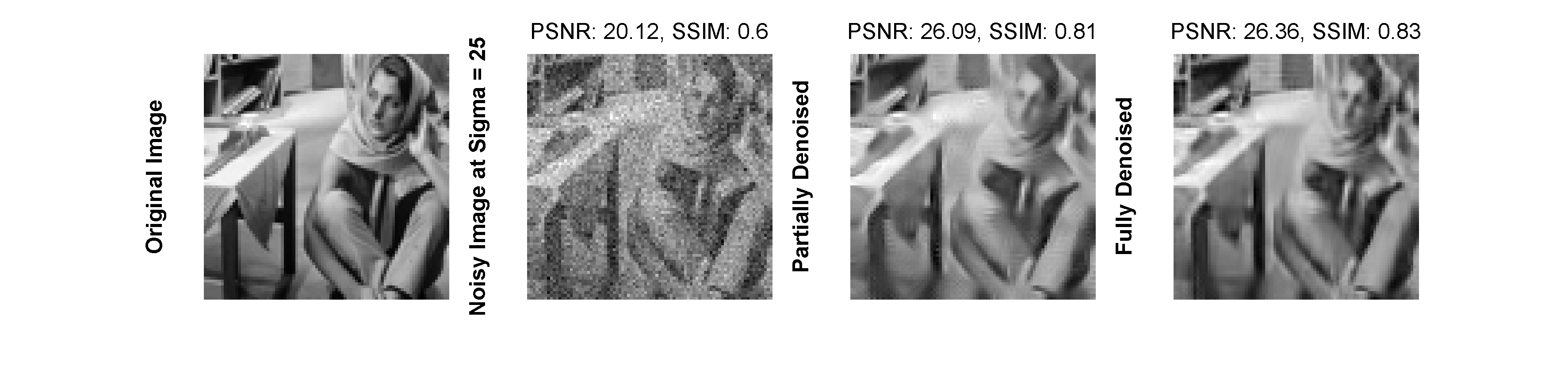}
\end{figure*}
\newpage
\begin{figure*}[t]
	\centering
	\includegraphics[width=1\linewidth]{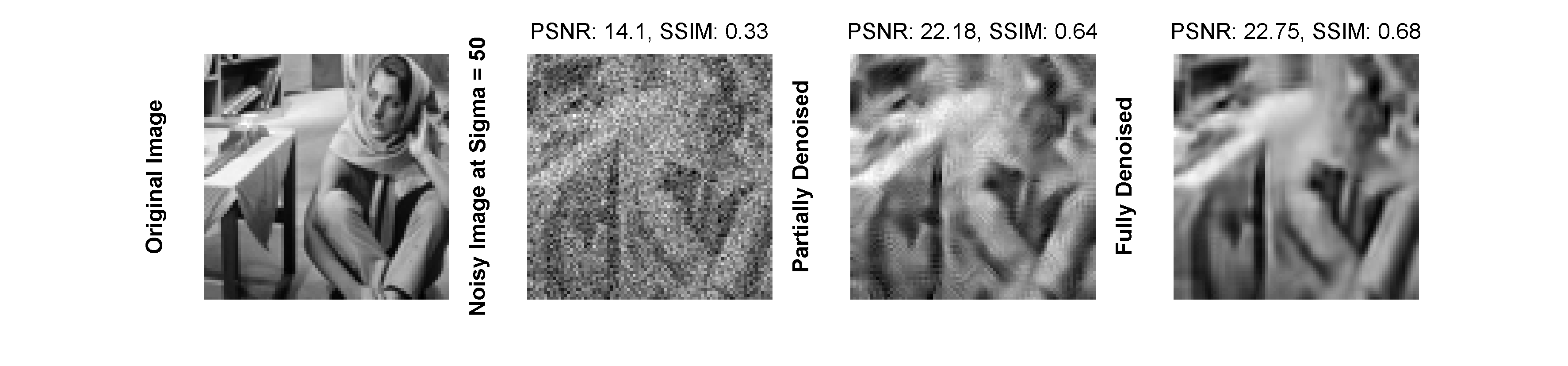}
	\includegraphics[width=1\linewidth]{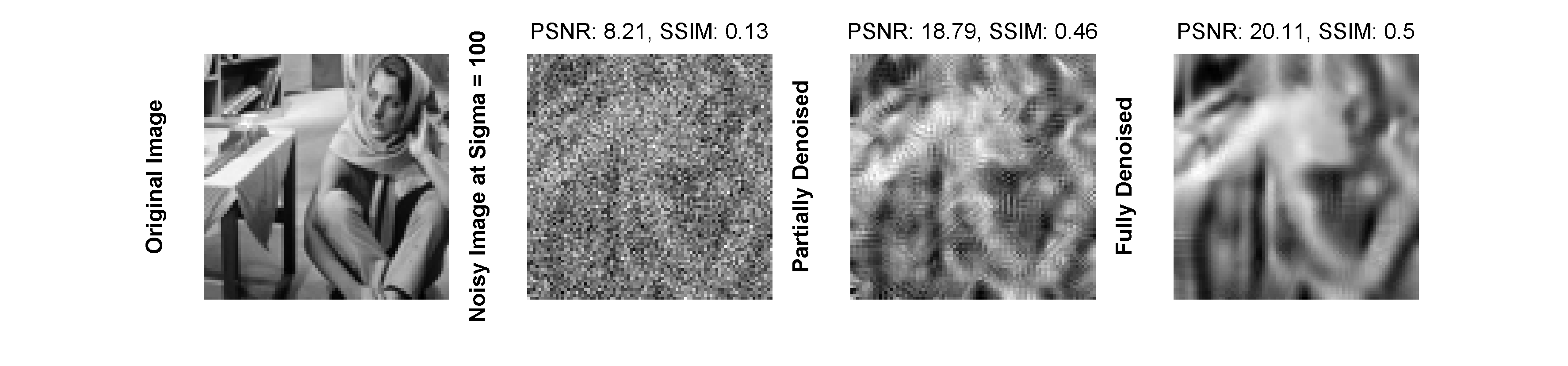}
	\includegraphics[width=1\linewidth]{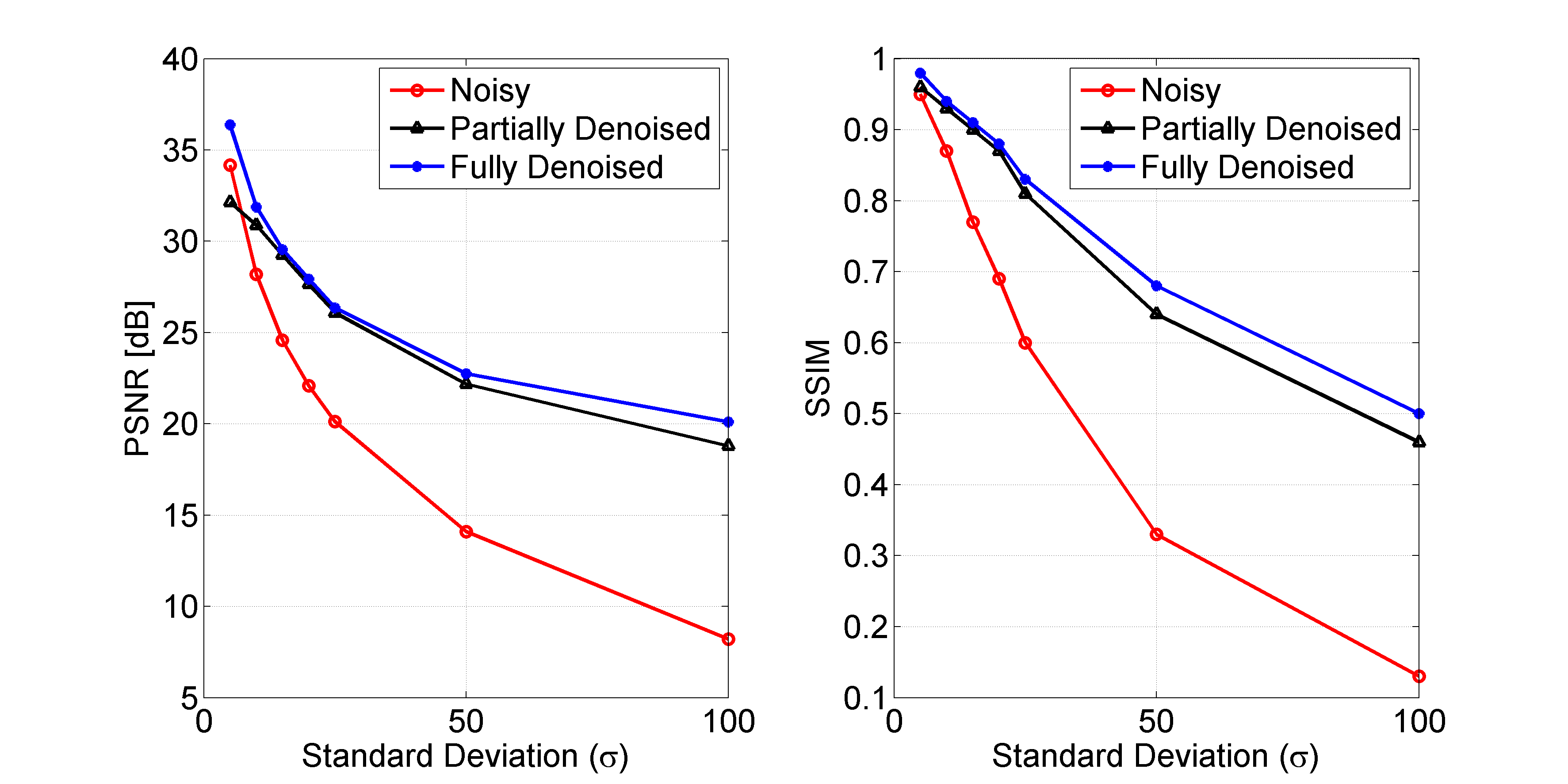}
	\caption{Denoising $256 \times256$ grayscale \textit{Barbara} standard test data images over noise $\sigma =  [5,10,15,20,25,50,100]$ when received at a node $\mu_\alpha$. Each row represent an original image, a noisy image, a partially denoised, and a fully denoised image, respectively, corrupted by a specific level of additive white Gaussian noise (AWGN). The graphical results in the end show PSNR [dB] and SSIM results in the form of graphs.}
\end{figure*}

\newpage
\begin{figure*}[t]
	\centering
	\includegraphics[width=1\linewidth]{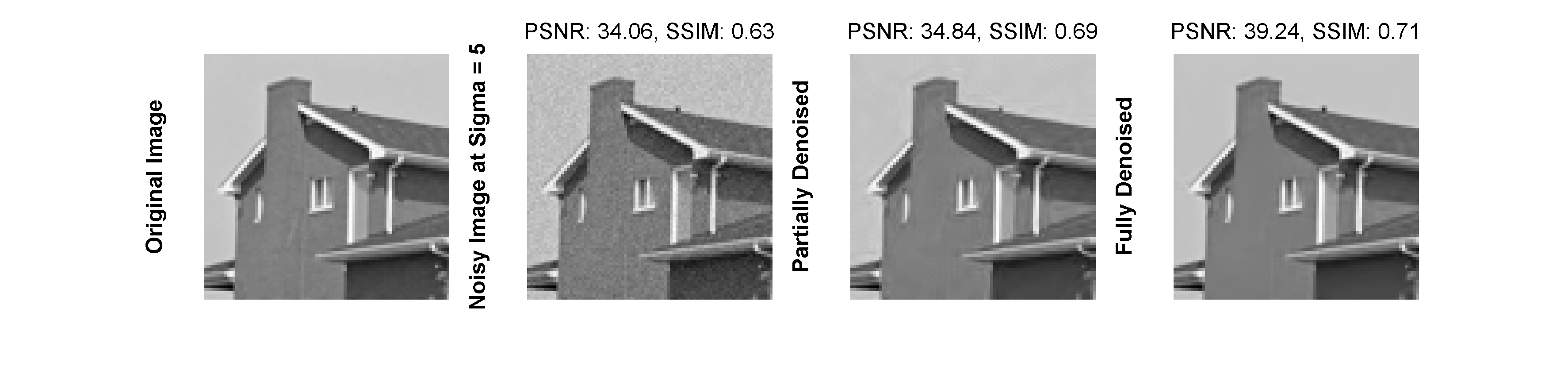}
	\includegraphics[width=1\linewidth]{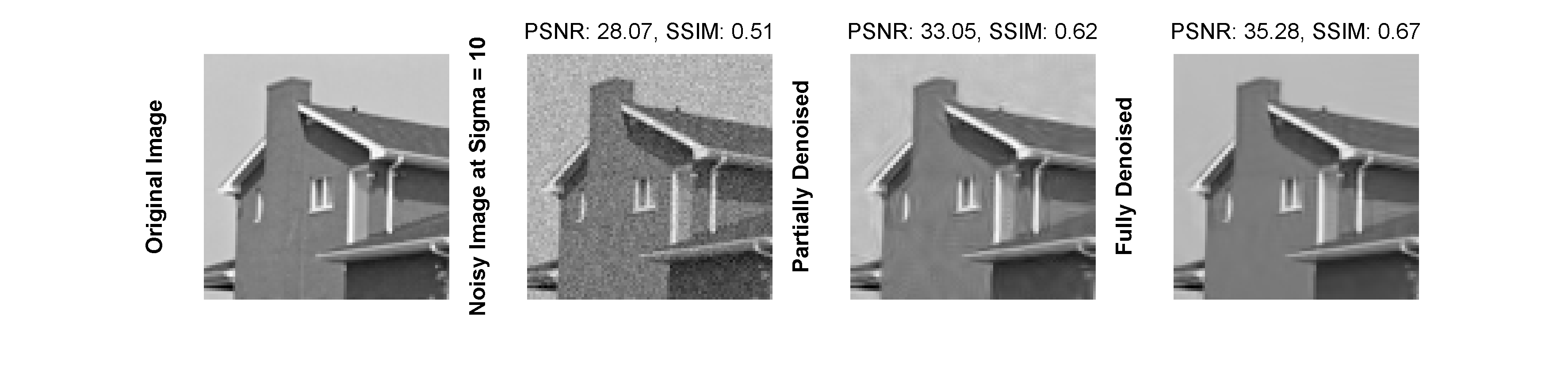}
	\includegraphics[width=1\linewidth]{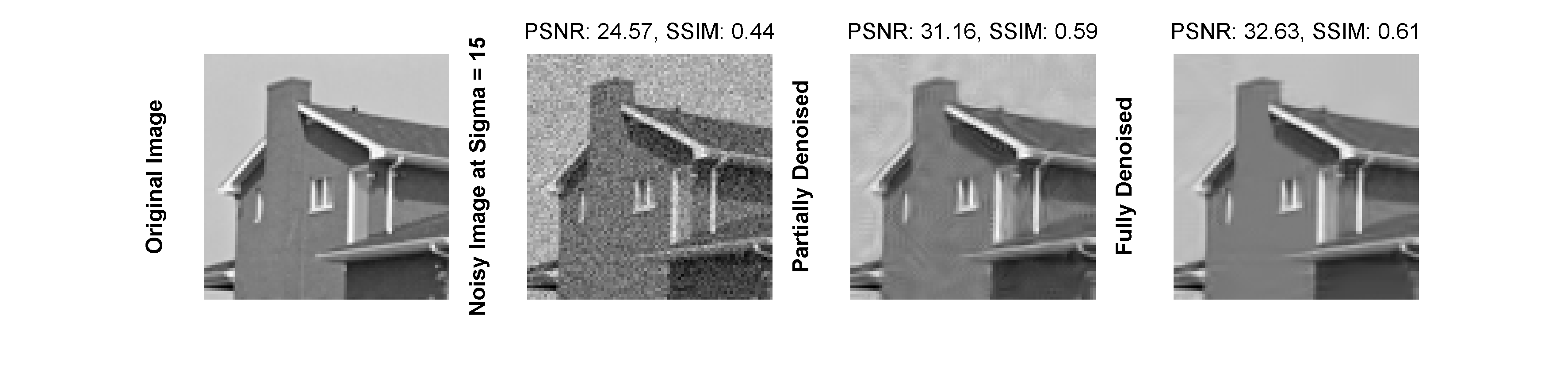}
	\includegraphics[width=1\linewidth]{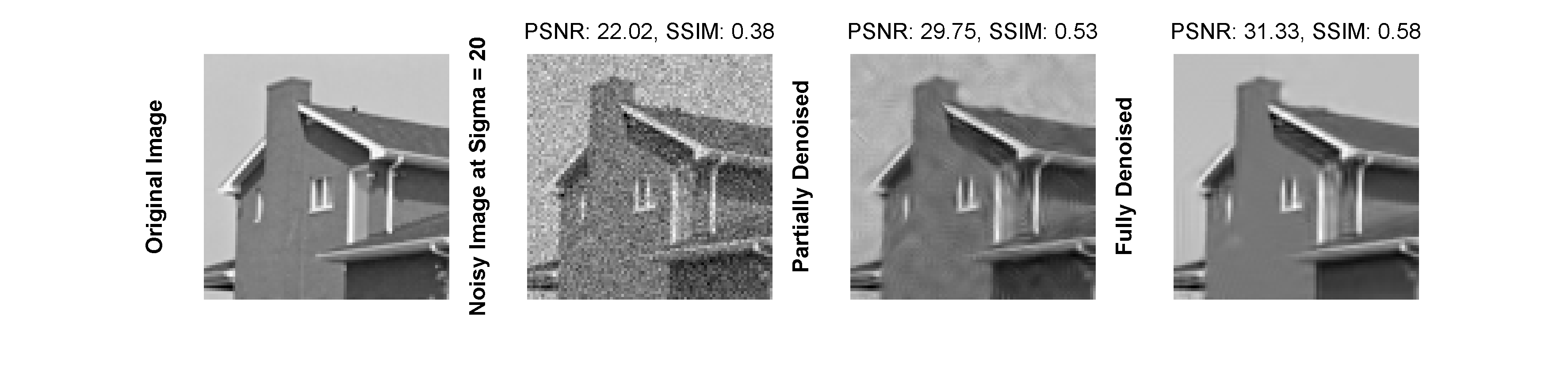}
	\includegraphics[width=1\linewidth]{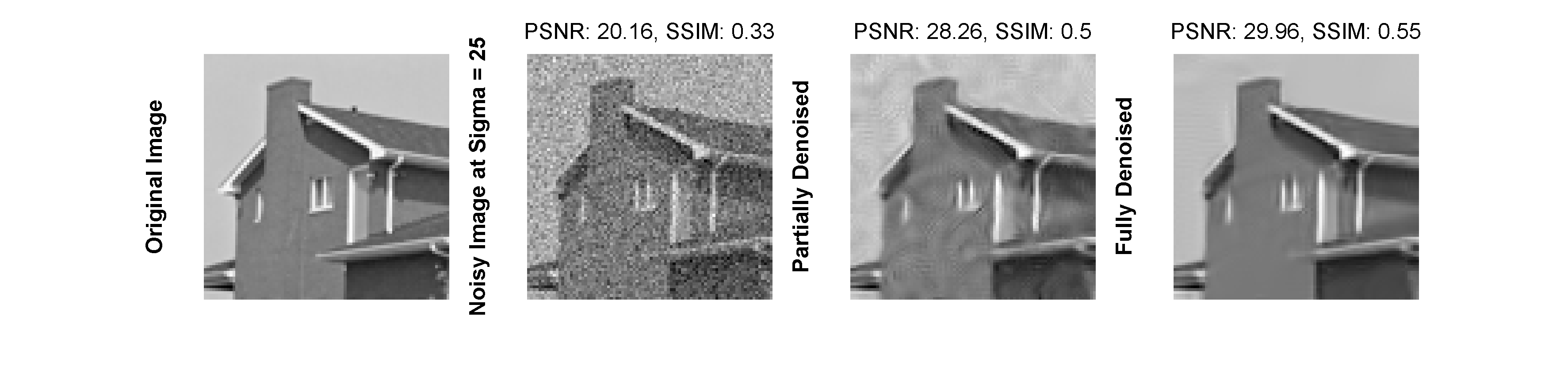}
\end{figure*}
\newpage
\begin{figure*}[t]
	\centering
	\includegraphics[width=1\linewidth]{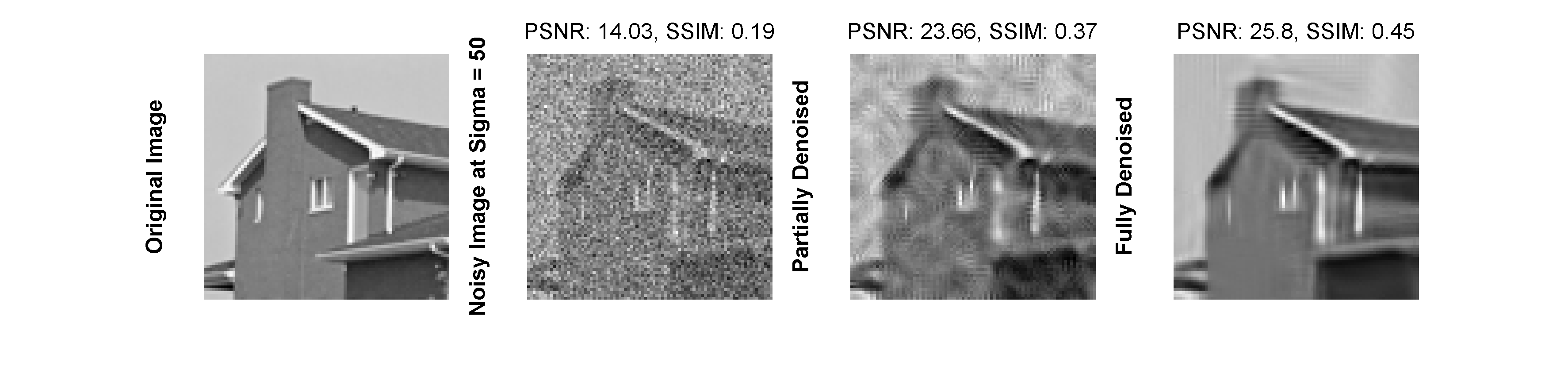}
	\includegraphics[width=1\linewidth]{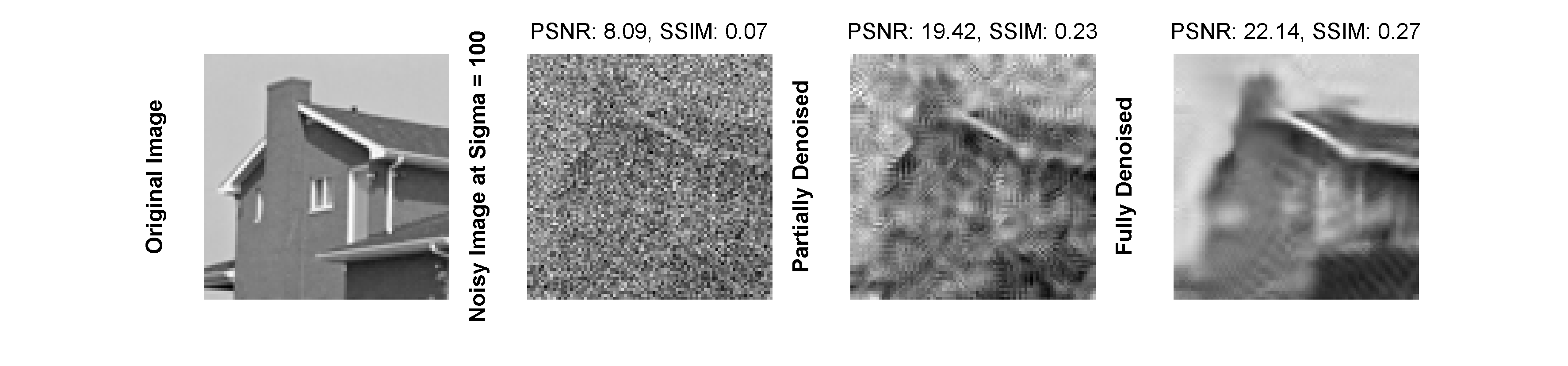}
	\includegraphics[width=1\linewidth]{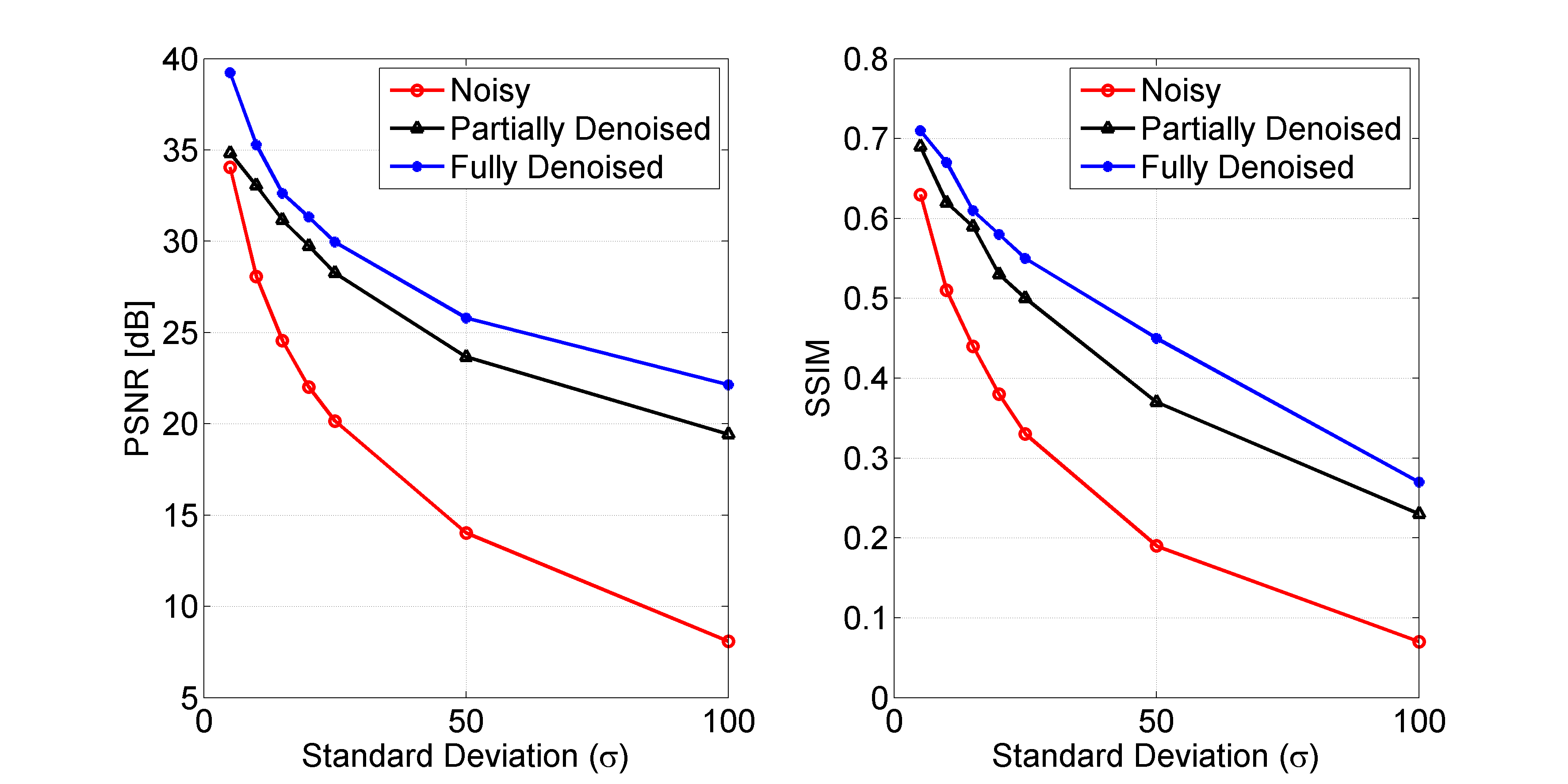}
	\caption{Denoising $256 \times256$ grayscale \textit{House} standard test data images over noise $\sigma =  [5,10,15,20,25,50,100]$ when received at a node $\mu_\alpha$. Each row represent an original image, a noisy image, a partially denoised, and a fully denoised image, respectively, corrupted by a specific level of additive white Gaussian noise (AWGN). The graphical results in the end show PSNR [dB] and SSIM results in the form of graphs.}
\end{figure*}

\newpage
\begin{figure*}[t]
	\centering
	\includegraphics[width=1\linewidth]{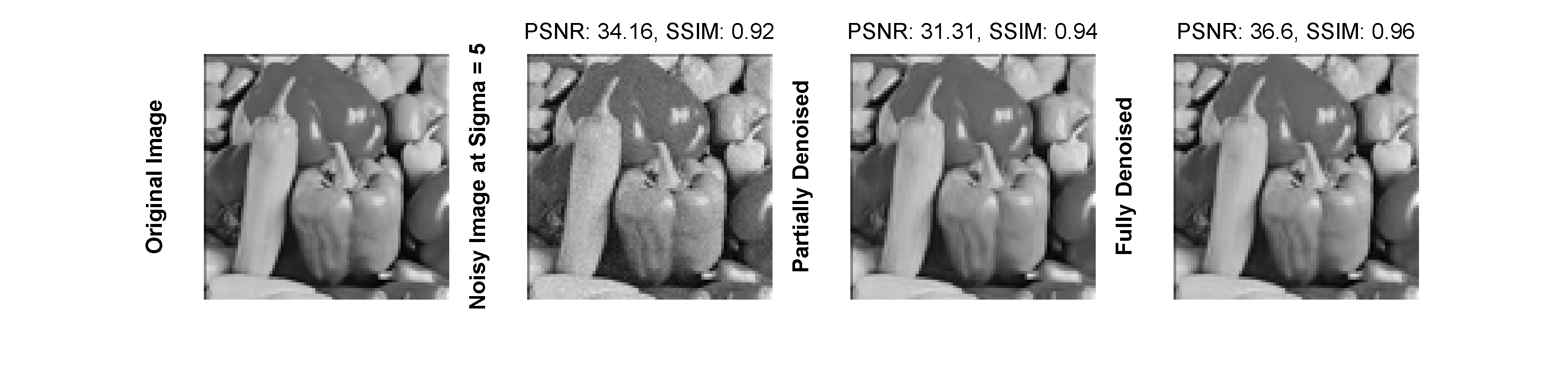}
	\includegraphics[width=1\linewidth]{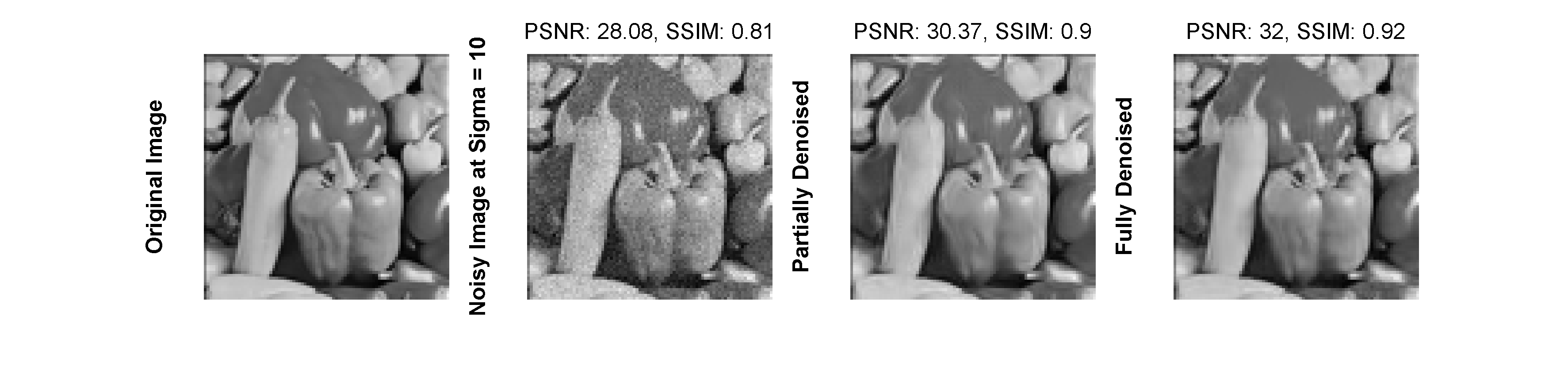}
	\includegraphics[width=1\linewidth]{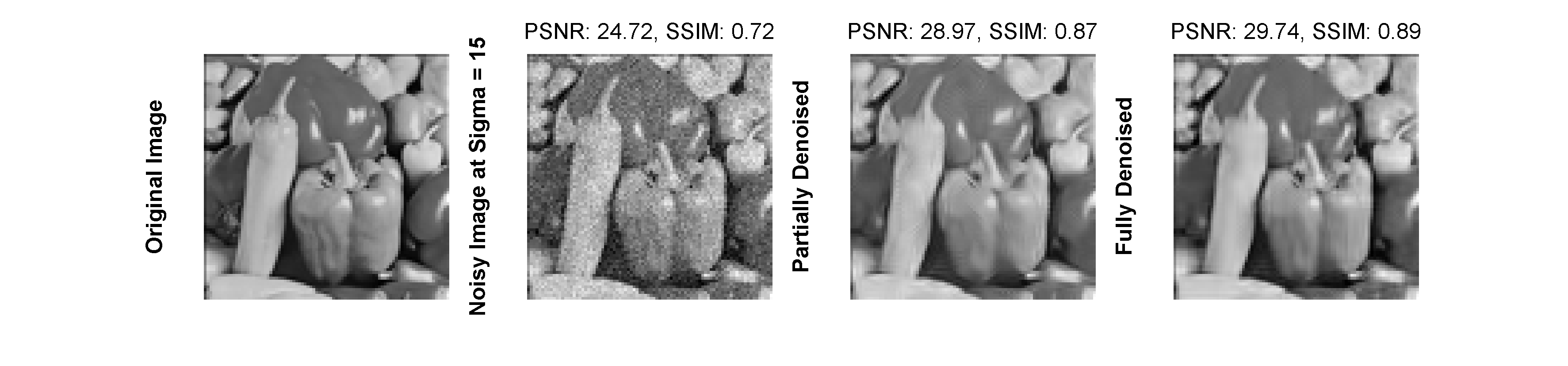}
	\includegraphics[width=1\linewidth]{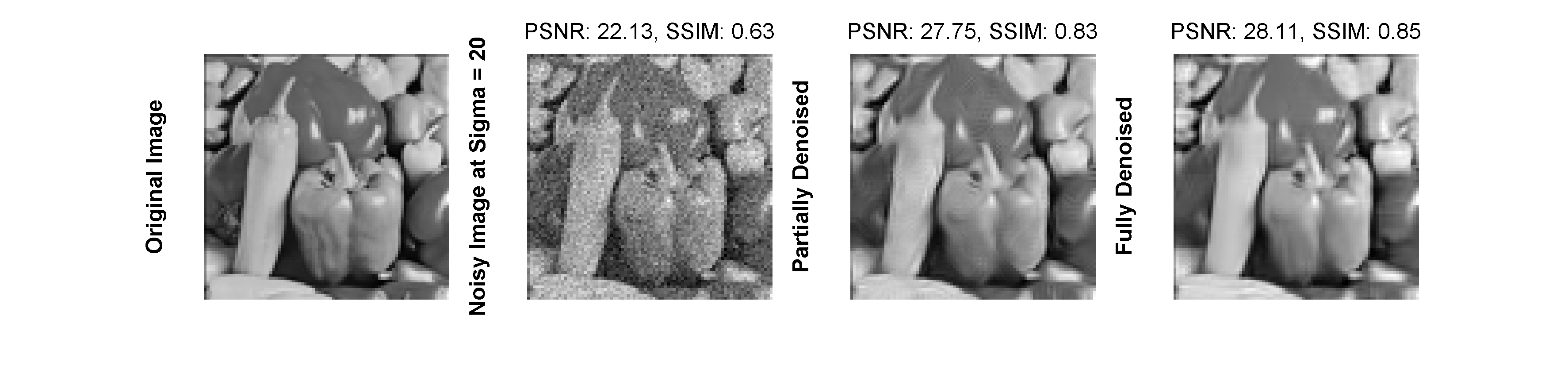}
	\includegraphics[width=1\linewidth]{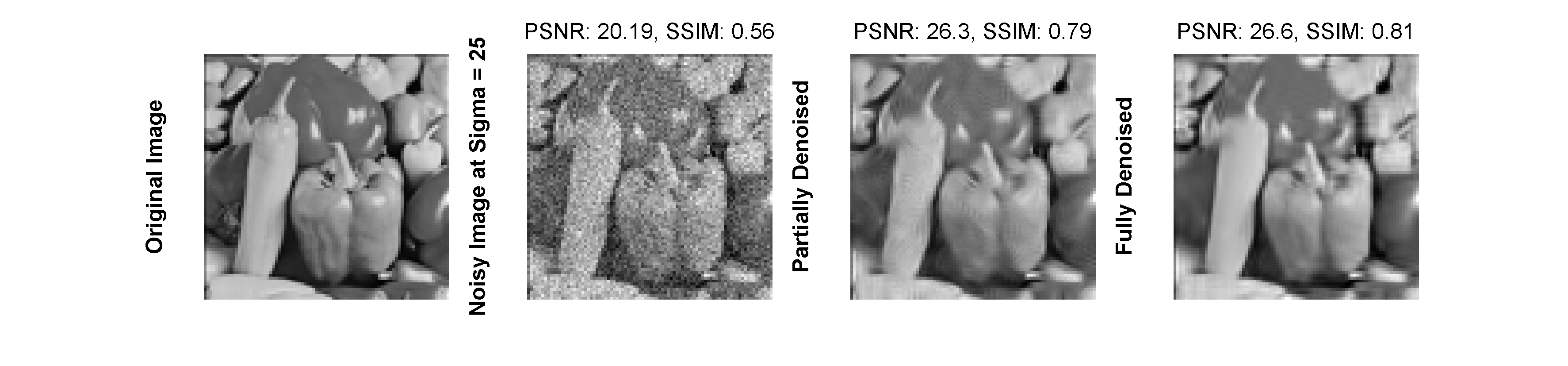}
\end{figure*}
\newpage
\begin{figure*}[t]
	\centering
	\includegraphics[width=1\linewidth]{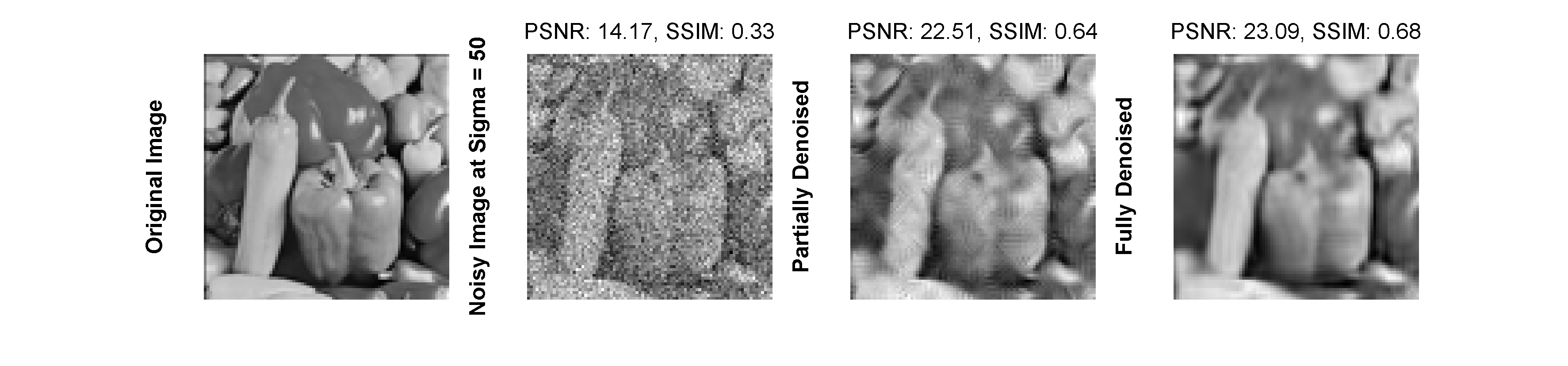}
	\includegraphics[width=1\linewidth]{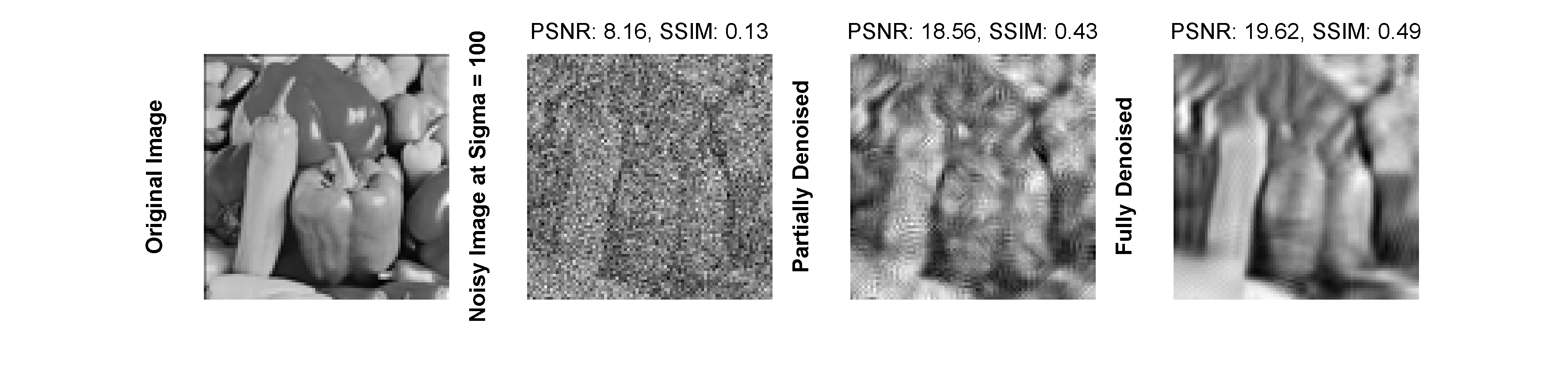}
	\includegraphics[width=1\linewidth]{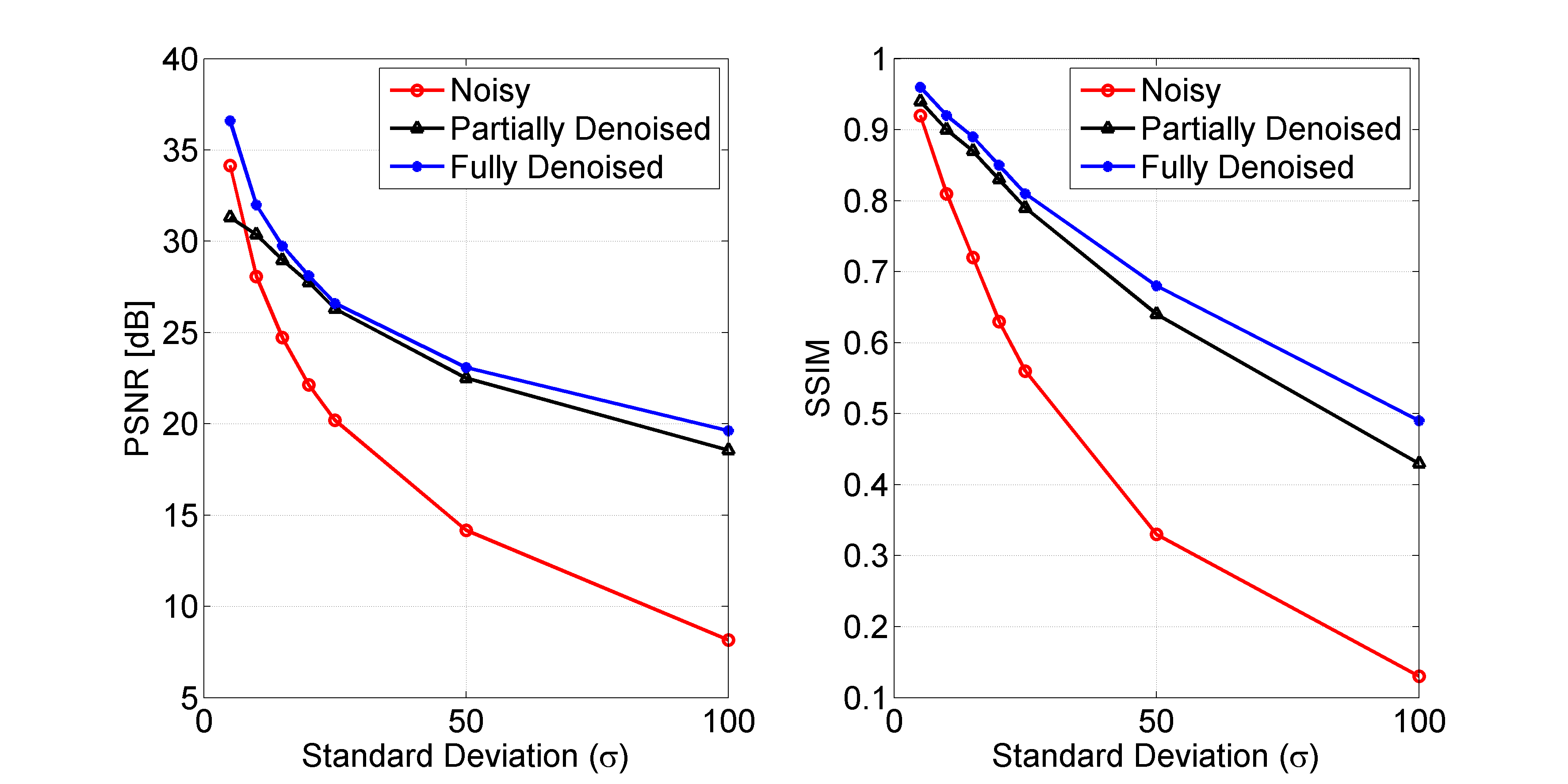}
	\caption{Denoising $256 \times256$ grayscale \textit{Peppers} standard test data images over noise $\sigma =  [5,10,15,20,25,50,100]$ when received at a node $\mu_\alpha$. Each row represent an original image, a noisy image, a partially denoised, and a fully denoised image, respectively, corrupted by a specific level of additive white Gaussian noise (AWGN). The graphical results in the end show PSNR [dB] and SSIM results in the form of graphs.}
\end{figure*}

\newpage
\begin{figure*}[t]
	\centering
	\includegraphics[width=1\linewidth]{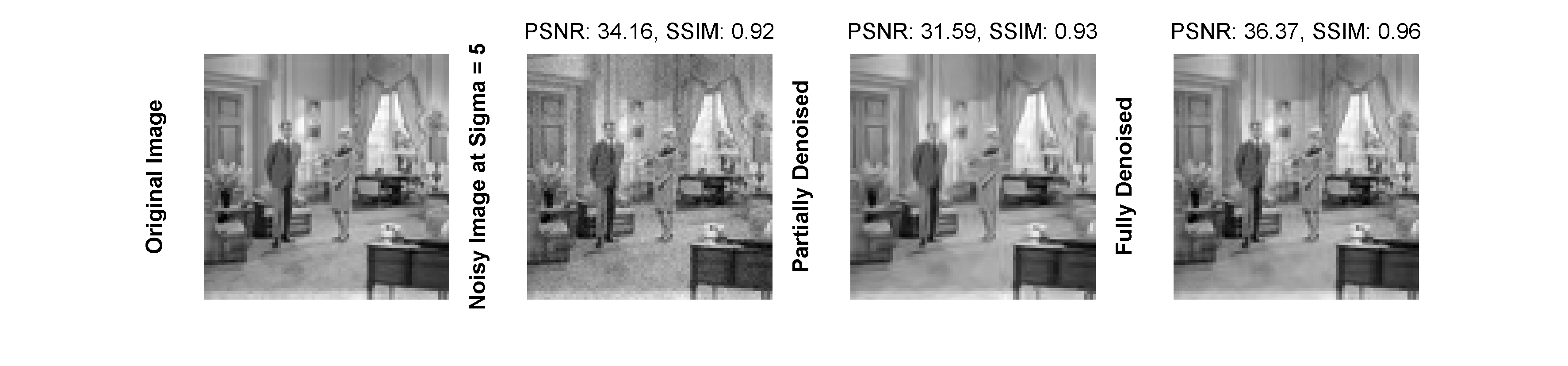}
	\includegraphics[width=1\linewidth]{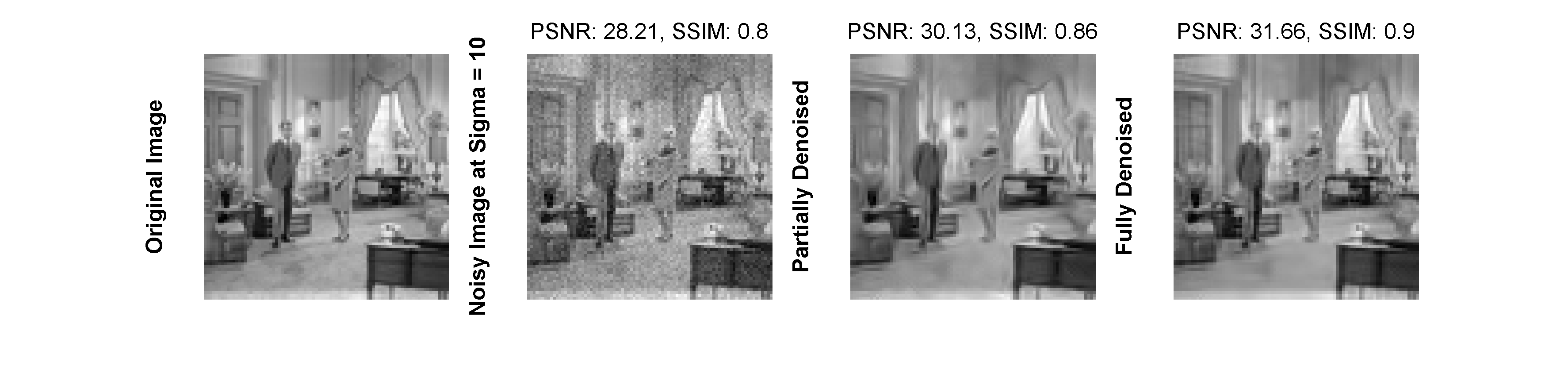}
	\includegraphics[width=1\linewidth]{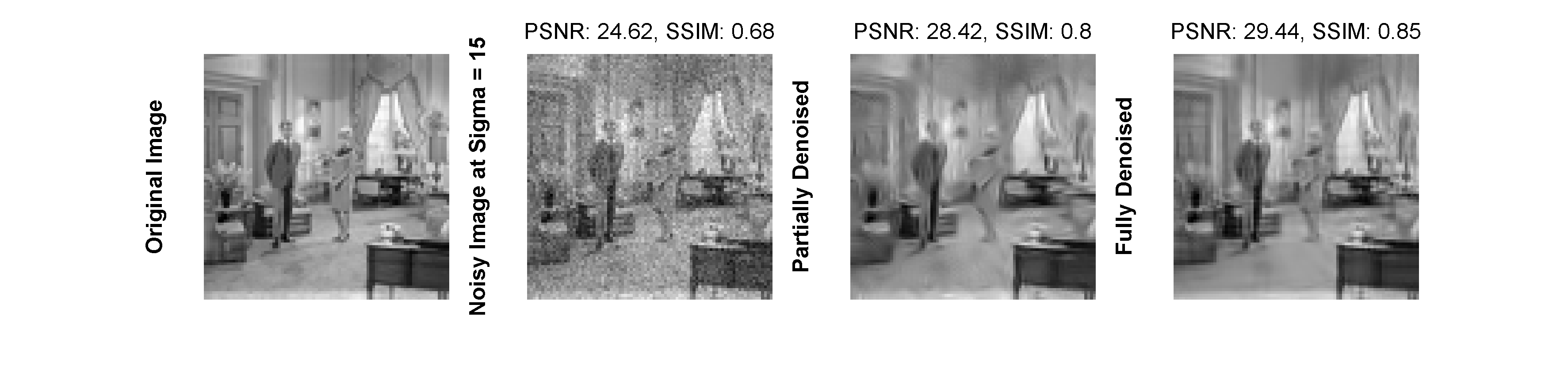}
	\includegraphics[width=1\linewidth]{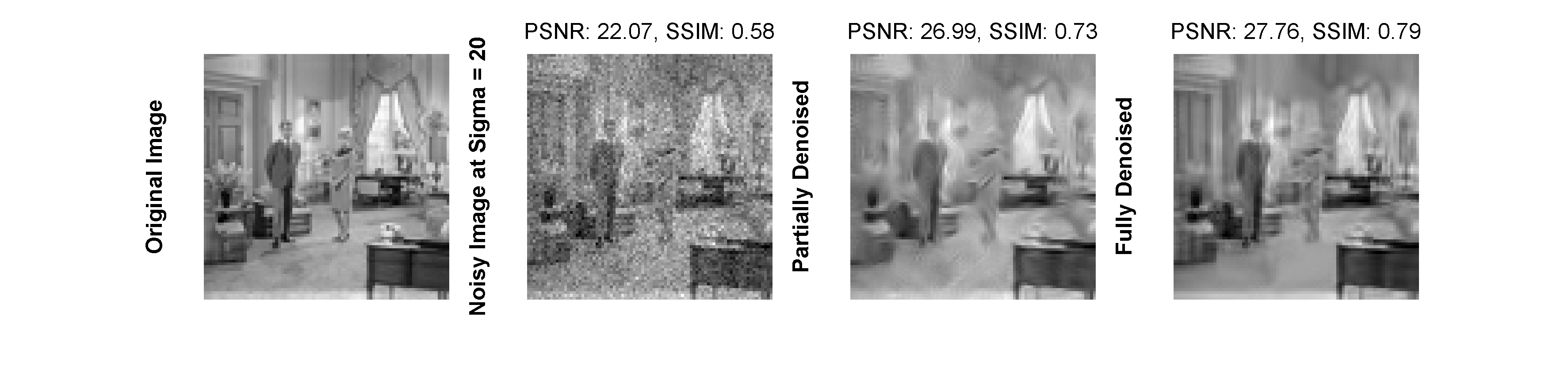}
	\includegraphics[width=1\linewidth]{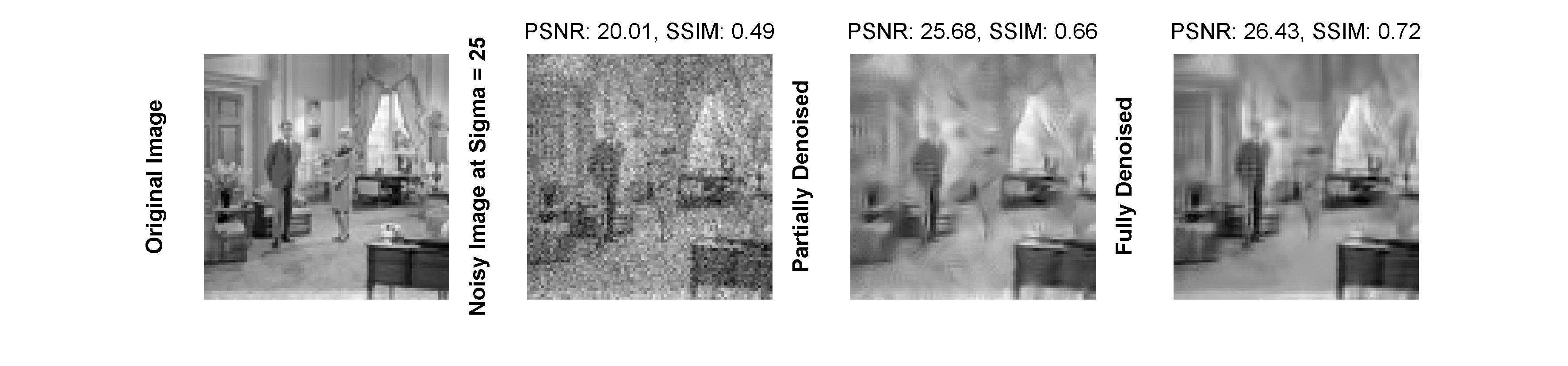}
\end{figure*}
\newpage
\begin{figure*}[t]
	\centering
	\includegraphics[width=1\linewidth]{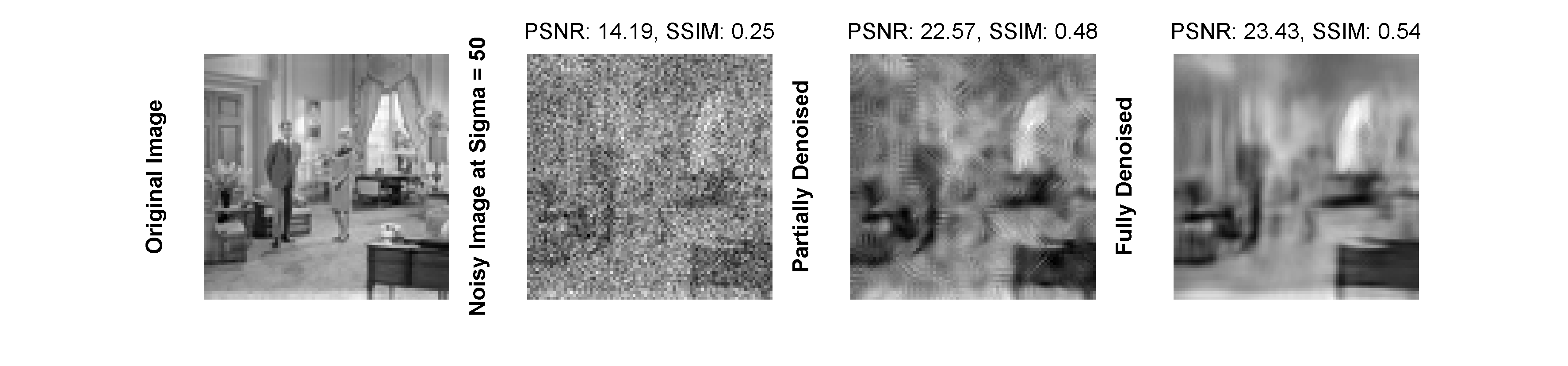}
	\includegraphics[width=1\linewidth]{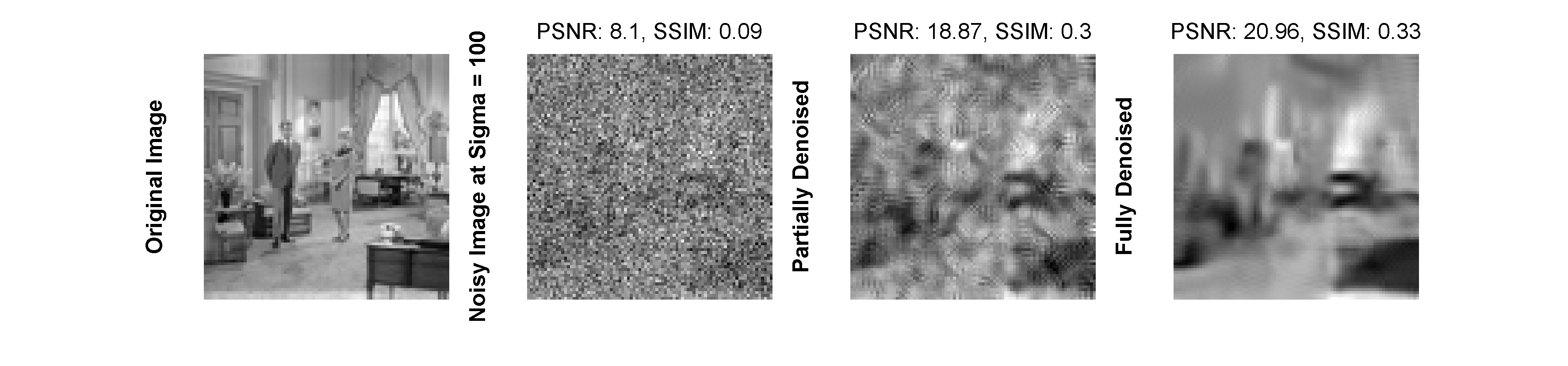}
	\includegraphics[width=1\linewidth]{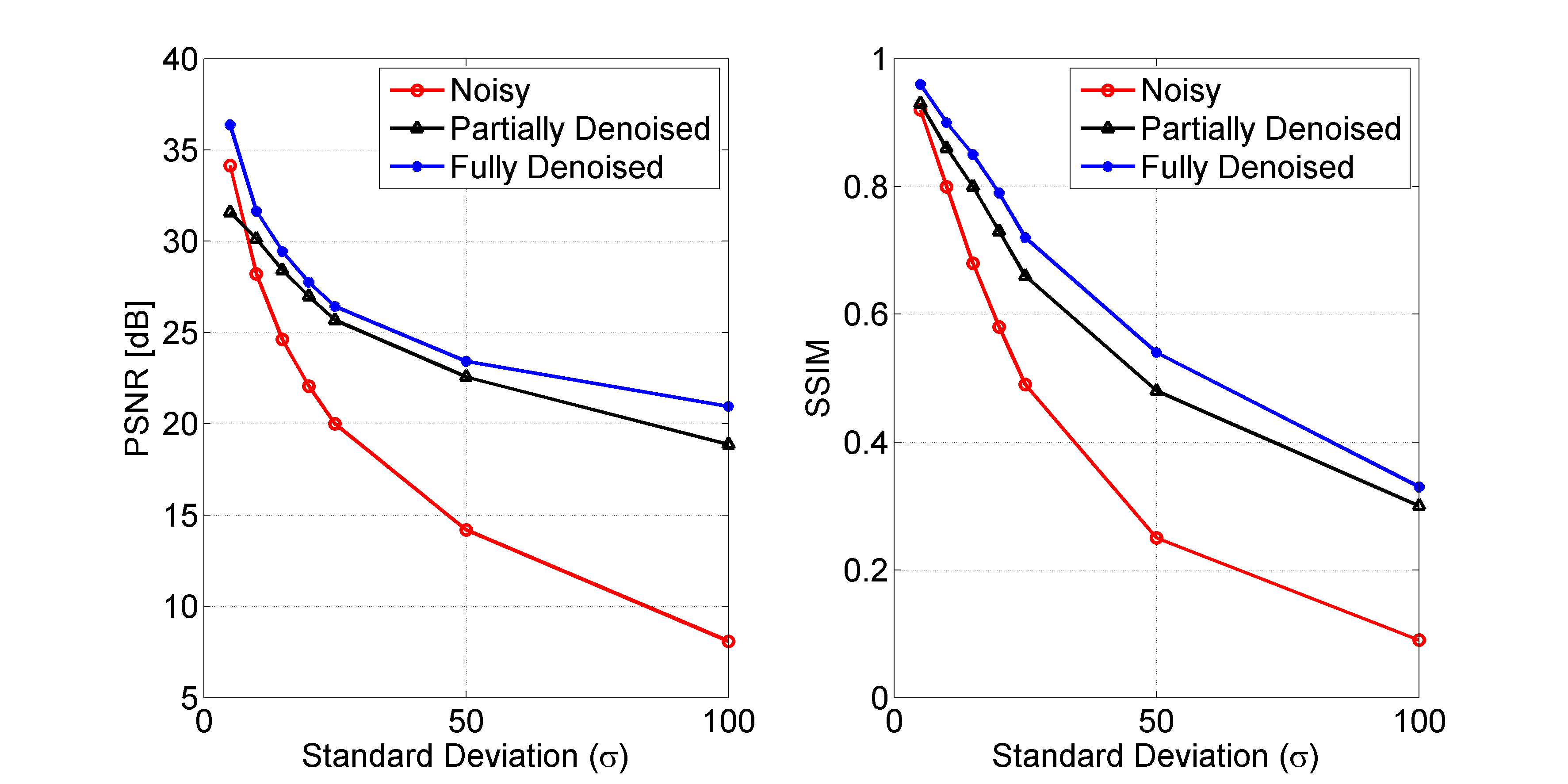}
	\caption{Denoising $256 \times256$ grayscale \textit{Living Room} standard test data images over noise $\sigma =  [5,10,15,20,25,50,100]$ when received at a node $\mu_\alpha$. Each row represent an original image, a noisy image, a partially denoised, and a fully denoised image, respectively, corrupted by a specific level of additive white Gaussian noise (AWGN). The graphical results in the end show PSNR [dB] and SSIM results in the form of graphs.}
\end{figure*}

\newpage
\begin{figure*}[t]
	\centering
	\includegraphics[width=1\linewidth]{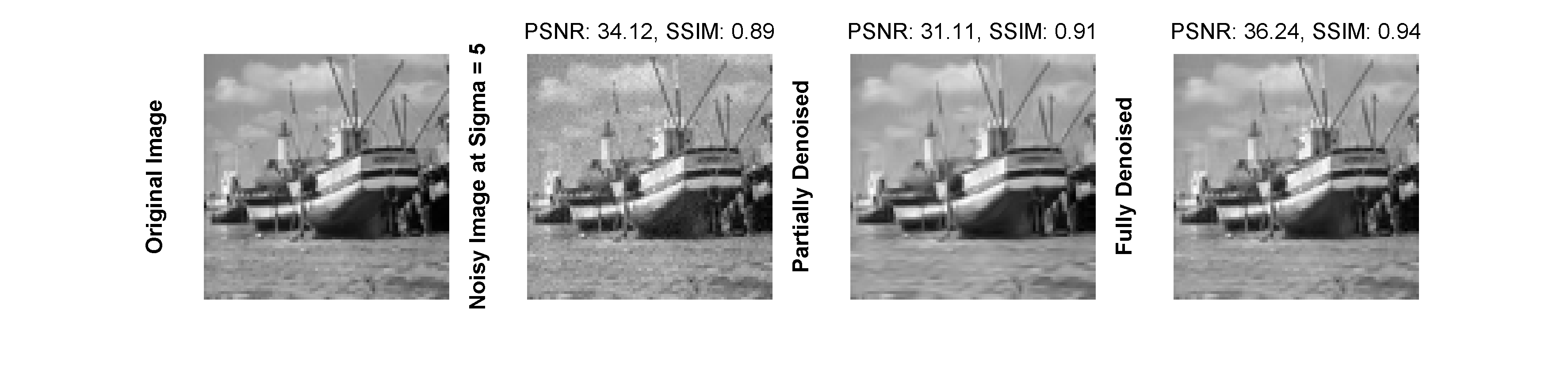}
	\includegraphics[width=1\linewidth]{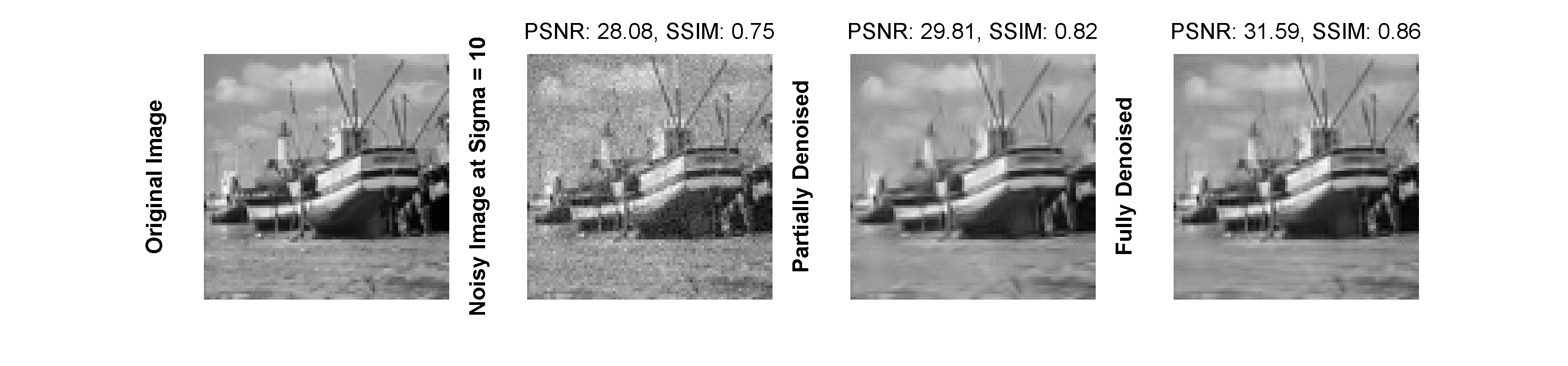}
	\includegraphics[width=1\linewidth]{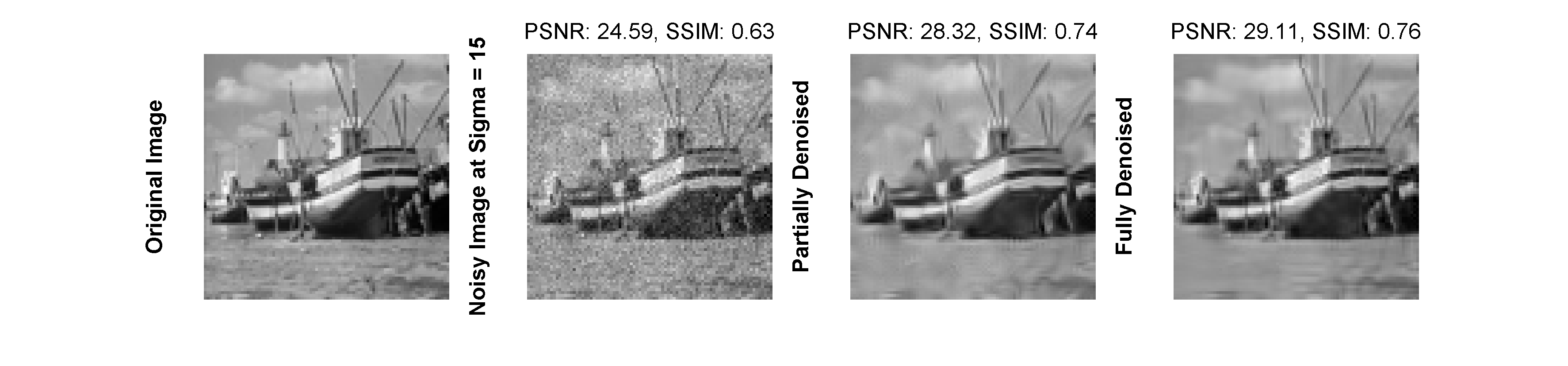}
	\includegraphics[width=1\linewidth]{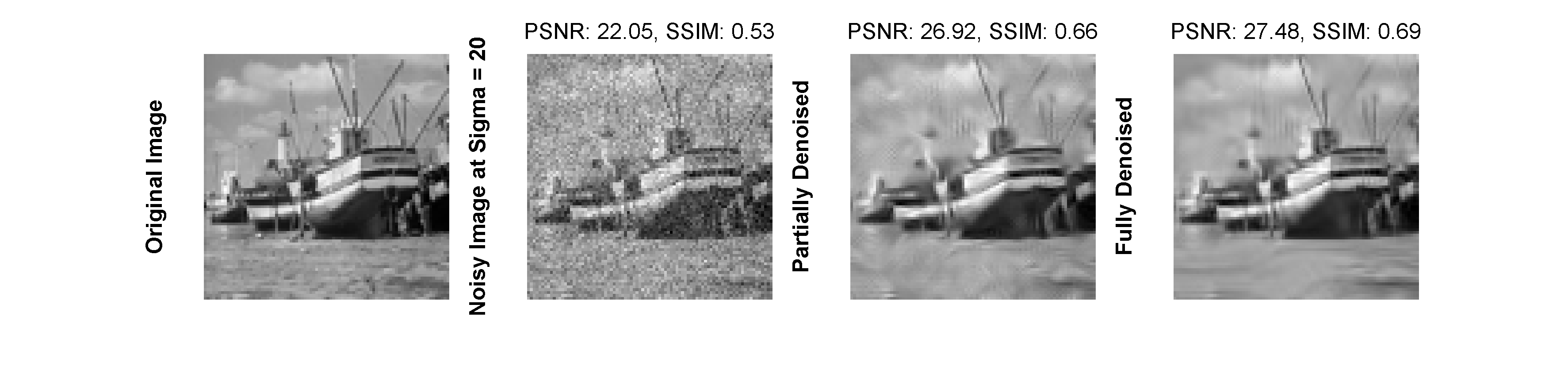}
	\includegraphics[width=1\linewidth]{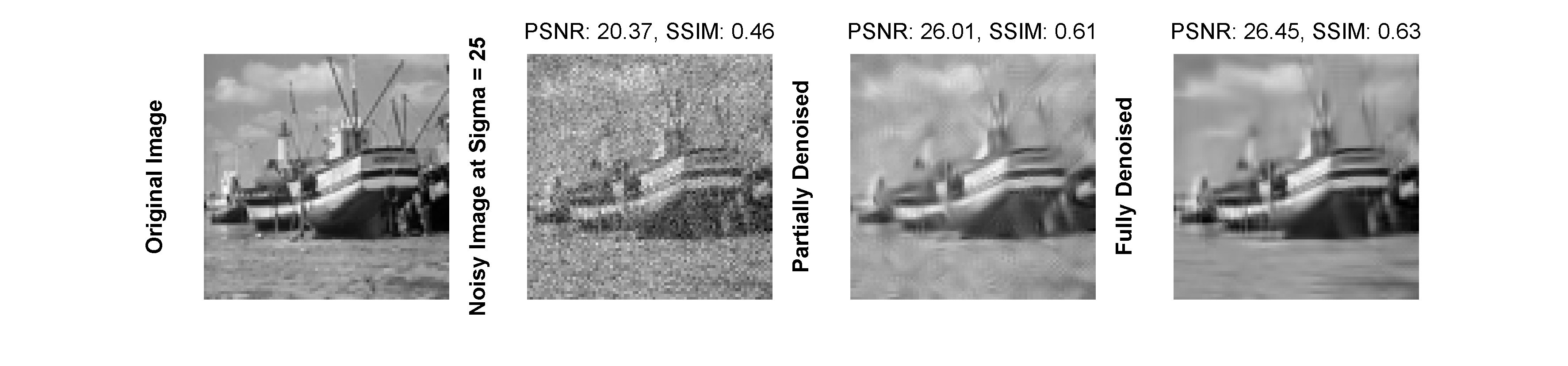}
\end{figure*}
\newpage
\begin{figure*}[t]
	\centering
	\includegraphics[width=1\linewidth]{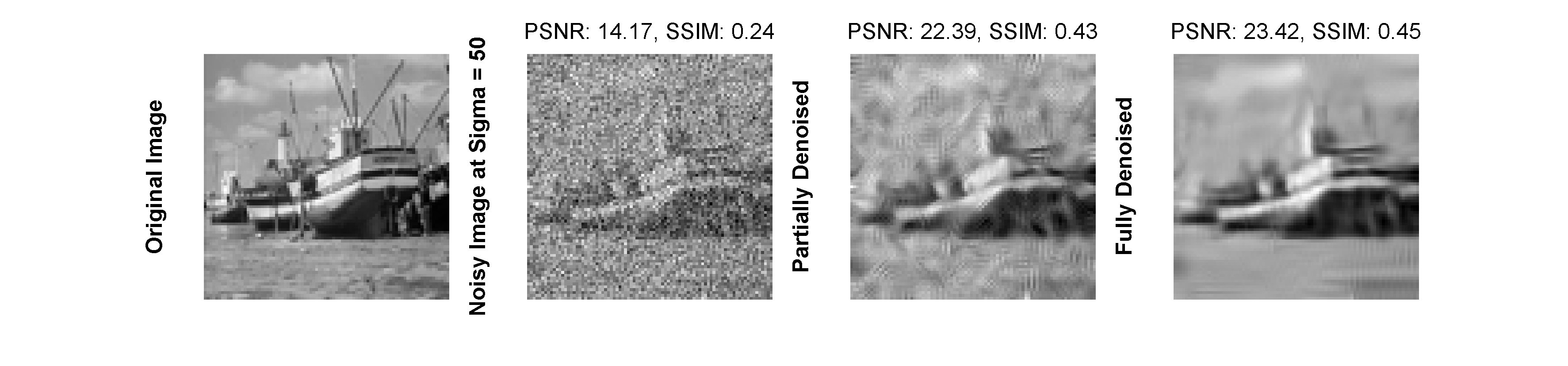}
	\includegraphics[width=1\linewidth]{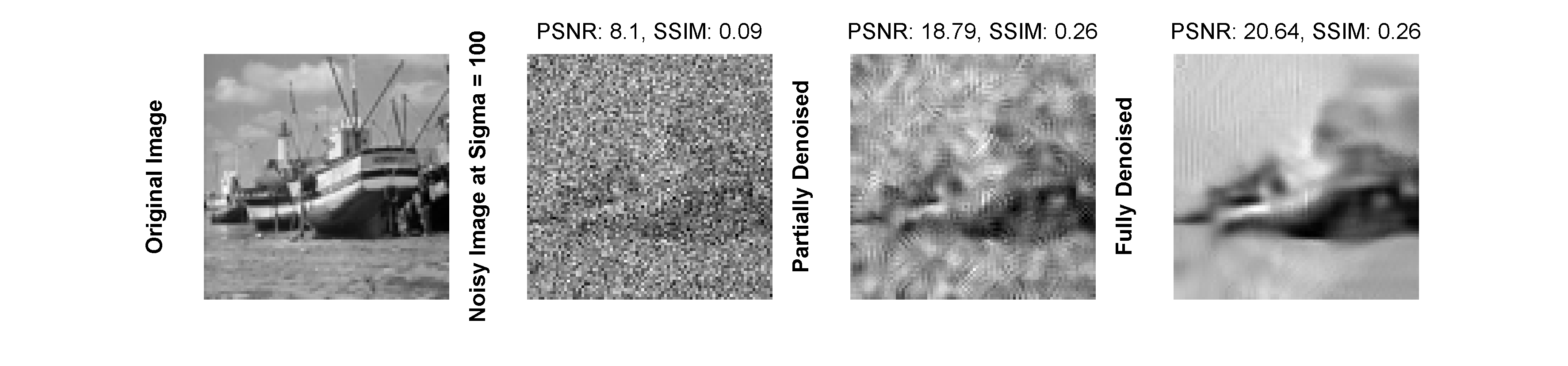}
	\includegraphics[width=1\linewidth]{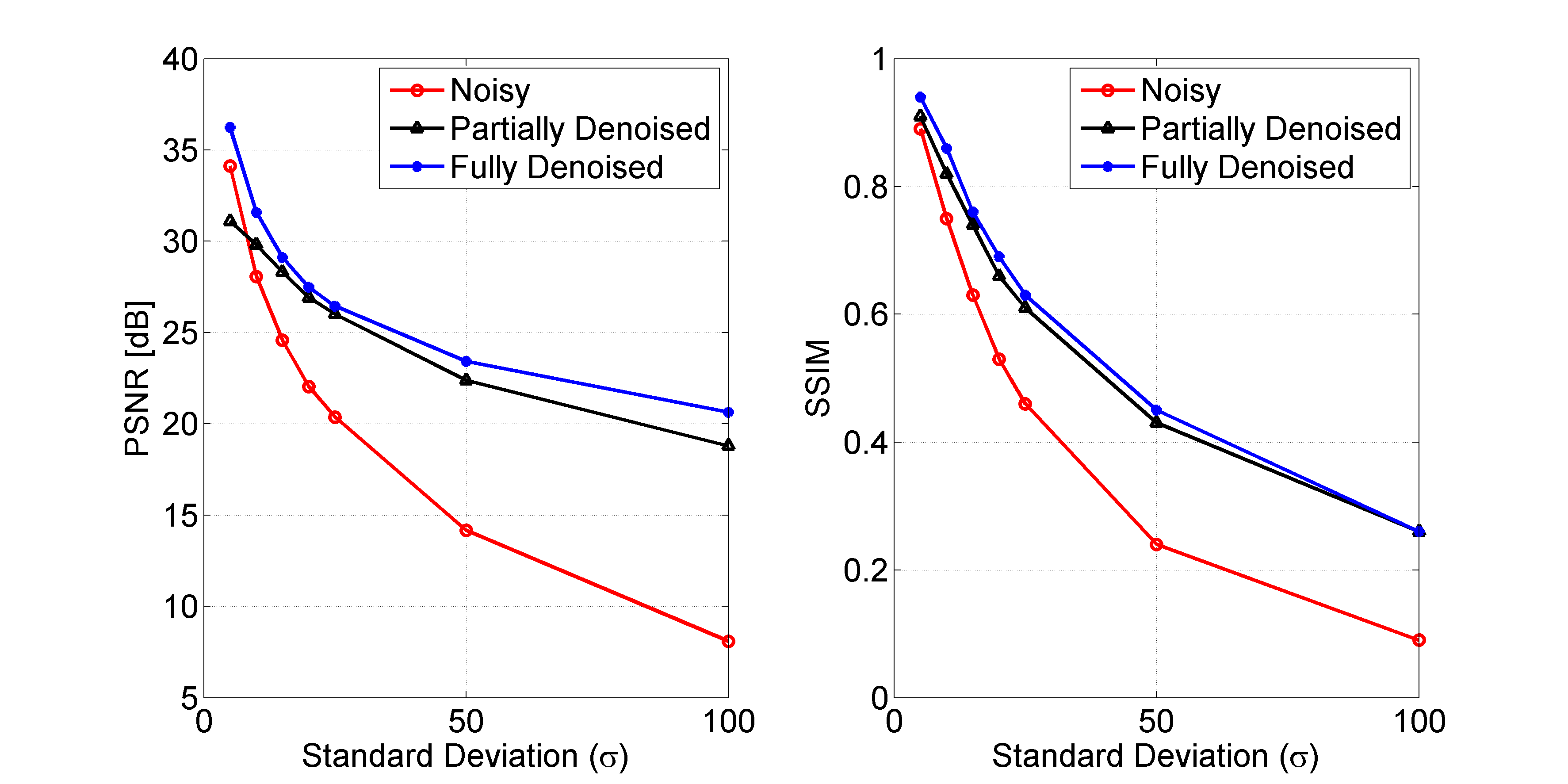}
	\caption{Denoising $256 \times256$ grayscale \textit{Boat} standard test data images over noise $\sigma =  [5,10,15,20,25,50,100]$ when received at a node $\mu_\alpha$. Each row represent an original image, a noisy image, a partially denoised, and a fully denoised image, respectively, corrupted by a specific level of additive white Gaussian noise (AWGN). The graphical results in the end show PSNR [dB] and SSIM results in the form of graphs.}
\end{figure*}

\newpage
\begin{figure*}[t]
	\centering
	\includegraphics[width=1\linewidth]{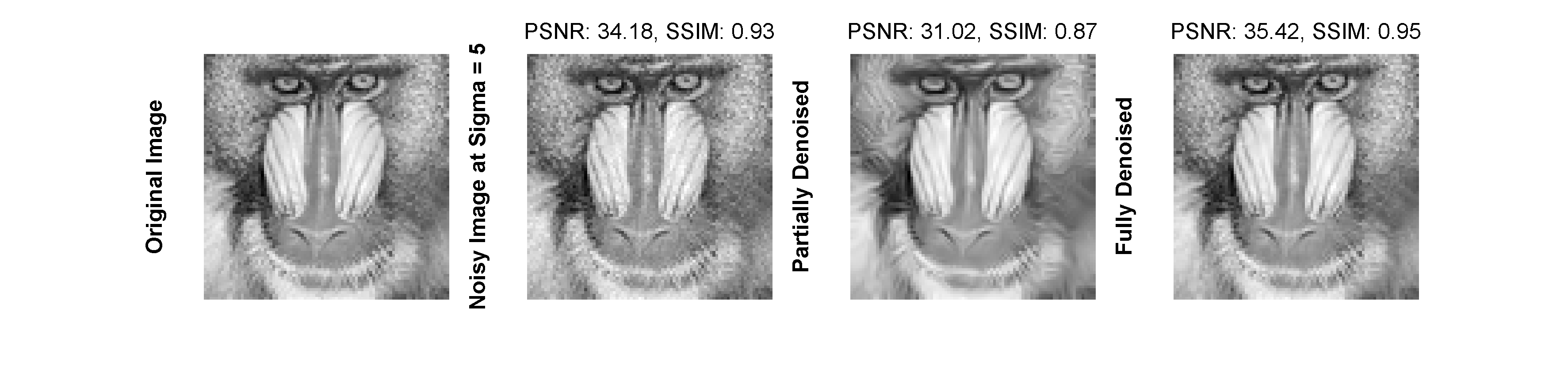}
	\includegraphics[width=1\linewidth]{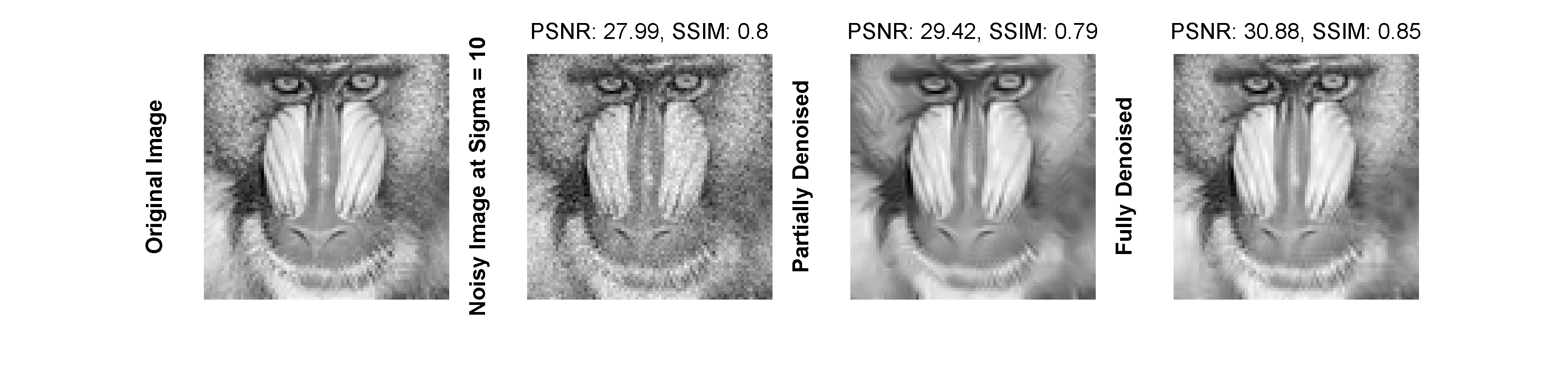}
	\includegraphics[width=1\linewidth]{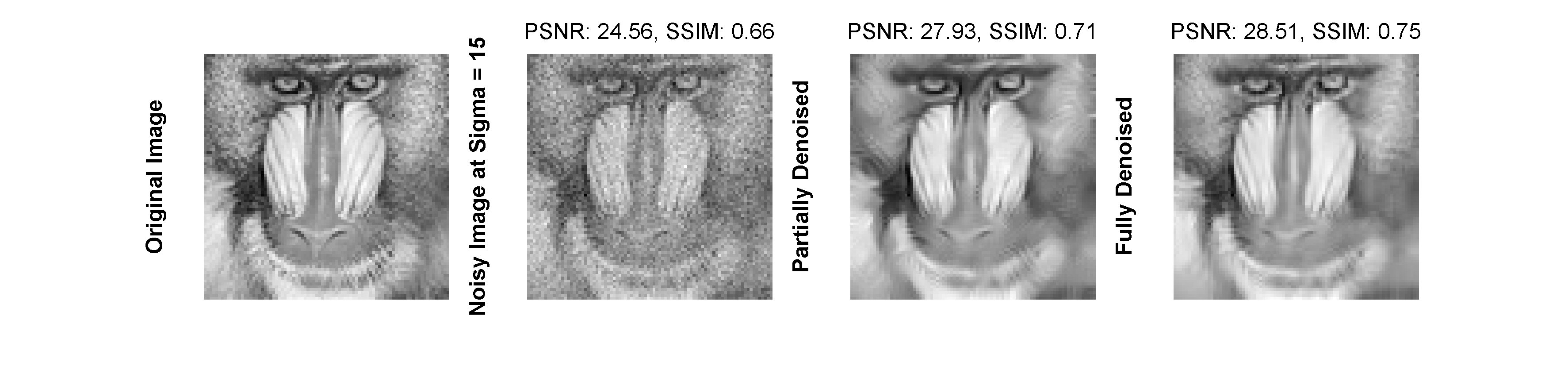}
	\includegraphics[width=1\linewidth]{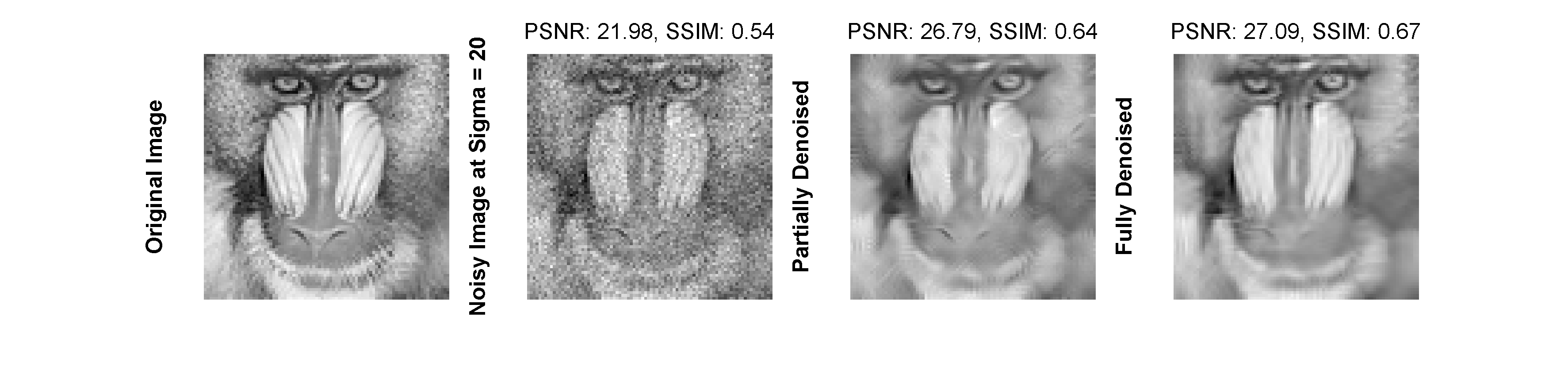}
	\includegraphics[width=1\linewidth]{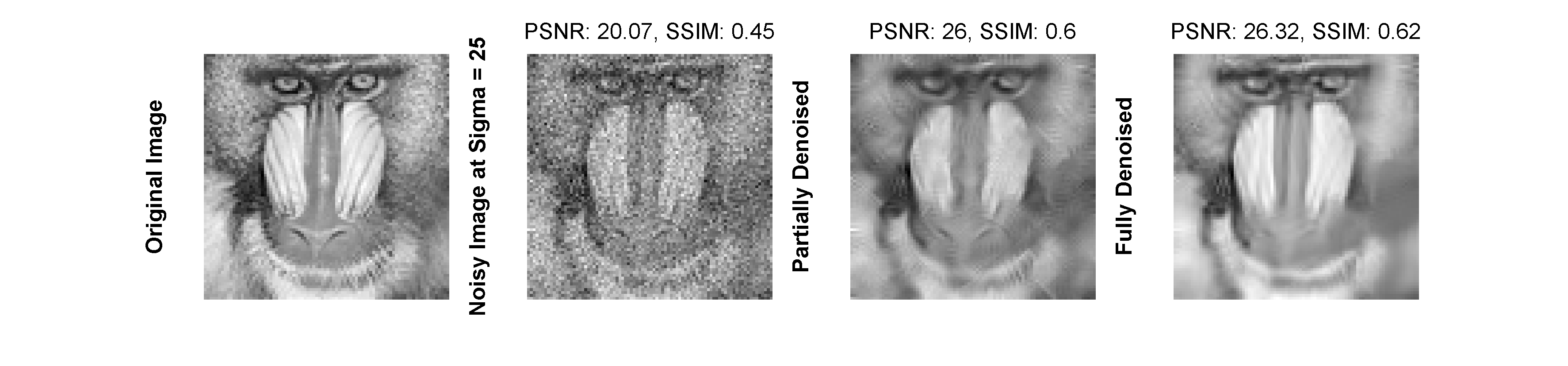}
\end{figure*}
\newpage
\begin{figure*}[t]
	\centering
	\includegraphics[width=1\linewidth]{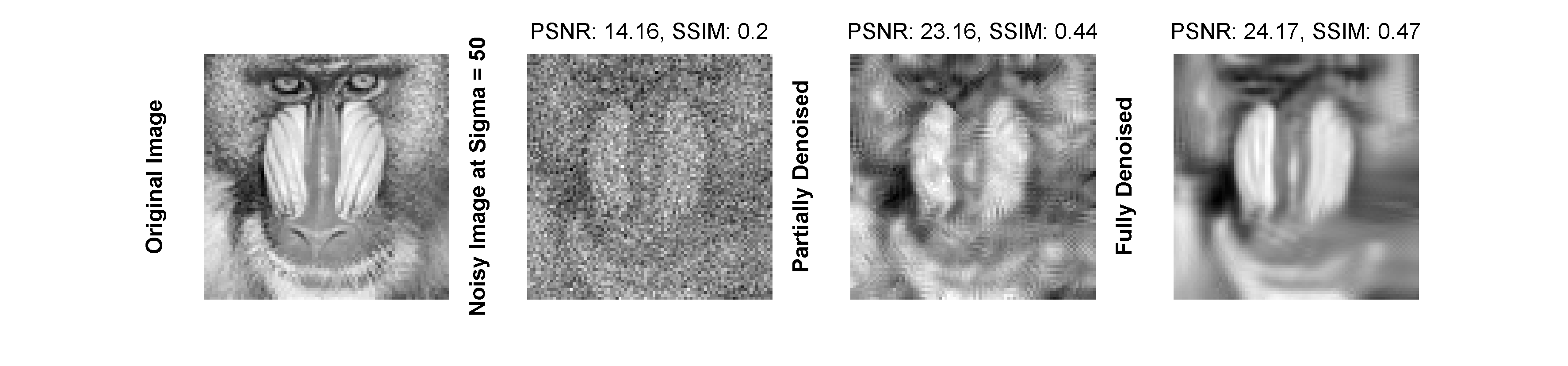}
	\includegraphics[width=1\linewidth]{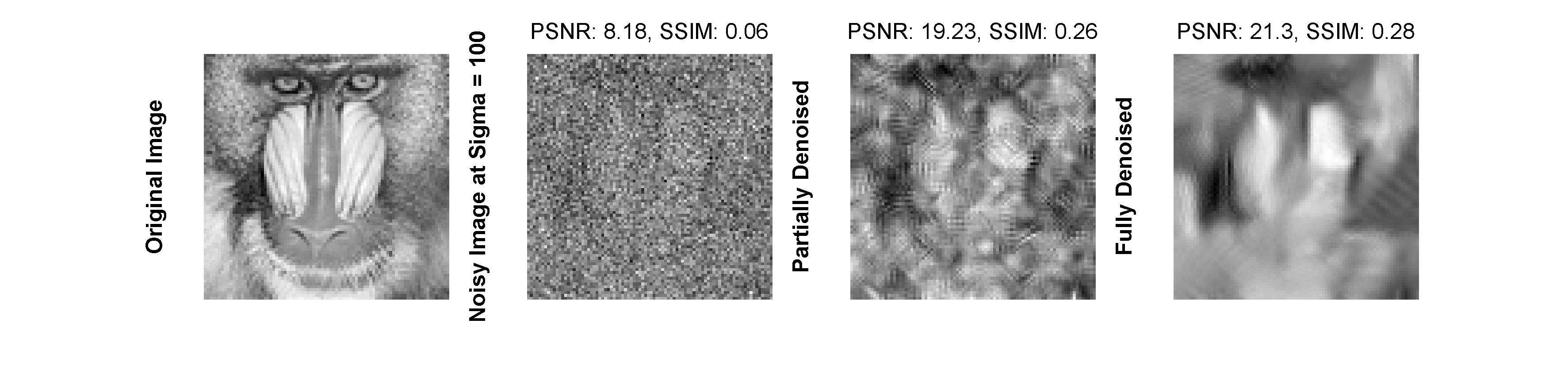}
	\includegraphics[width=1\linewidth]{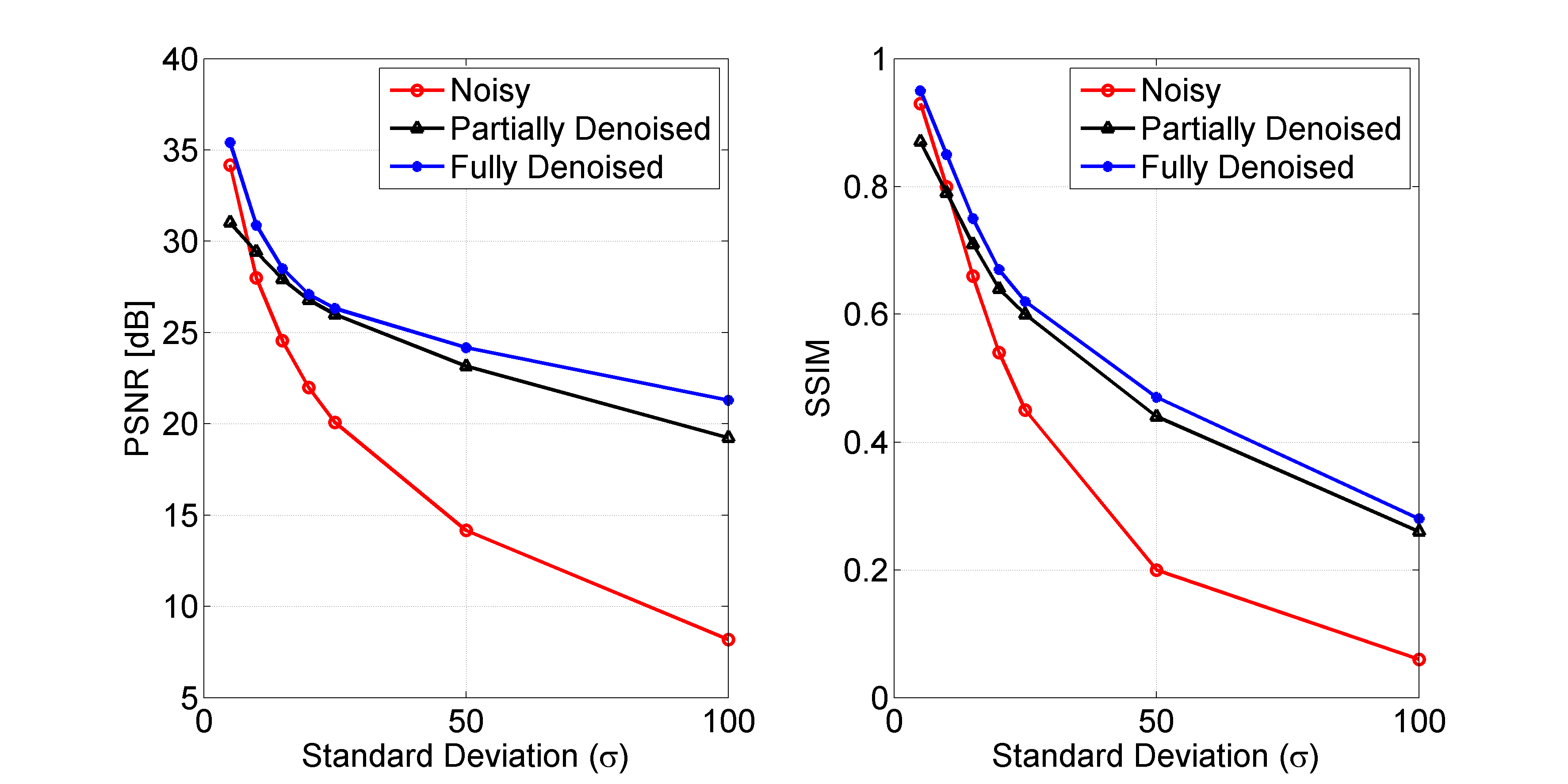}
	\caption{Denoising $256 \times256$ grayscale \textit{Mandrill} standard test data images over noise $\sigma =  [5,10,15,20,25,50,100]$ when received at a node $\mu_\alpha$. Each row represent an original image, a noisy image, a partially denoised, and a fully denoised image, respectively, corrupted by a specific level of additive white Gaussian noise (AWGN). The graphical results in the end show PSNR [dB] and SSIM results in the form of graphs.}
\end{figure*}

\clearpage
\bibliographystyle{ieeetr}
\bibliography{Behzad}

\begin{thebibliography}{10}

\bibitem{7110295}
Z.~Sheng, C.~Mahapatra, C.~Zhu, and V.~C.~M. Leung, ``Recent advances in
  industrial wireless sensor networks toward efficient management in iot,''
  {\em IEEE Access}, vol.~3, pp.~622--637, 2015.

\bibitem{6604028}
S.~Marchesani, L.~Pomante, F.~Santucci, and M.~Pugliese, ``Demo abstract: A
  cryptographic scheme for real-world wireless sensor networks applications,''
  in {\em ACM/IEEE International Conference on Cyber-Physical Systems (ICCPS)},
  pp. 249-249, April 2013.

\bibitem{7887698}
G.~Mois, S.~Folea, and T.~Sanislav, ``Analysis of three iot-based wireless
  sensors for environmental monitoring,'' {\em IEEE Transactions on
  Instrumentation and Measurement}, vol.~66, pp.~2056--2064, Aug 2017.

\bibitem{7808887}
T.~Kageyama, M.~Miura, A.~Maeda, A.~Mori, and S.~S. Lee, ``A wireless sensor
  network platform for water quality monitoring,'' in {\em IEEE SENSORS}, pp.
  1-3, Oct 2016.

\bibitem{7528266}
C.~Grumazescu, V.~A. Vlăduţă, and G.~Subaşu, ``Wsn solutions for
  communication challenges in military live simulation environments,'' in {\em
  International Conference on Communications (COMM)}, pp. 319-322, June 2016.

\bibitem{7997348}
P.~Sun, A.~Boukerche, and Q.~Wu, ``Theoretical analysis of the target detection
  rules for the uav-based wireless sensor networks,'' in {\em IEEE
  International Conference on Communications (ICC)}, pp. 1-6, May 2017.

\bibitem{Yick20082292}
J.~Yick, B.~Mukherjee, and D.~Ghosal, ``Wireless sensor network survey,'' {\em
  Computer Networks}, vol.~52, no.~12, pp.~2292--2330, 2008.

\bibitem{6883890}
O.~Salem, Y.~Liu, and A.~Mehaoua, ``Anomaly detection in medical wsns using
  enclosing ellipse and chi-square distance,'' in {\em IEEE International
  Conference on Communications (ICC)}, pp 3658-3663, June 2014.

\bibitem{5474813}
M.~r.~Akhondi, A.~Talevski, S.~Carlsen, and S.~Petersen, ``Applications of
  wireless sensor networks in the oil, gas and resources industries,'' in {\em
  24th IEEE International Conference on Advanced Information Networking and
  Applications (AINA)}, pp. 941-948, April 2010.

\bibitem{Bharucha:2008:DPL:1461796.1461800}
Z.~Bharucha and H.~Haas, ``The distribution of path losses for uniformly
  distributed nodes in a circle,'' {\em Rec. Lett. Commun.}, vol.~2008,
  pp.~4:1--4:4, Jan. 2008.

\bibitem{7368292}
N.~Sibeko, P.~Mudali, O.~Oki, and A.~Alaba, ``Performance evaluation of routing
  protocols in uniform and normal node distributions using inter-mesh wireless
  networks,'' in {\em World Symposium on Computer Networks and Information
  Security (WSCNIS)}, pp. 1-6, Sept 2015.

\bibitem{heinzelman2000application}
W.~B. Heinzelman, {\em Application-specific protocol architectures for wireless
  networks}.
\newblock PhD thesis, Massachusetts Institute of Technology, 2000.

\bibitem{7016049}
M.~Behzad, N.~Javaid, A.~Sana, M.~A. Khan, N.~Saeed, Z.~A. Khan, and U.~Qasim,
  ``Tsddr: Threshold sensitive density controlled divide and rule routing
  protocol for wireless sensor networks,'' in {\em Ninth International
  Conference on Broadband and Wireless Computing, Communication and
  Applications}, pp. 78-83, Nov 2014.

\bibitem{7920983}
M.~Behzad and Y.~Ge, ``Performance optimization in wireless sensor networks: A
  novel collaborative compressed sensing approach,'' in {\em IEEE 31st
  International Conference on Advanced Information Networking and Applications
  (AINA)}, pp. 749-756, March 2017.

\bibitem{behzad2017distributed}
M.~Behzad, M.~Javaid, M.~Parahca, and S.~Khan, ``Distributed pca and consensus
  based energy efficient routing protocol for wsns.,'' {\em Journal of
  Information Science \& Engineering}, vol.~33, no.~5, 2017.

\bibitem{qaisar2013compressive}
S.~Qaisar, R.~M. Bilal, W.~Iqbal, M.~Naureen, and S.~Lee, ``Compressive
  sensing: From theory to applications, a survey,'' {\em Journal of
  Communications and networks}, vol.~15, no.~5, pp.~443--456, 2013.

\bibitem{5454399}
W.~U. Bajwa, J.~Haupt, A.~M. Sayeed, and R.~Nowak, ``Compressed channel
  sensing: A new approach to estimating sparse multipath channels,'' {\em
  Proceedings of the IEEE}, vol.~98, pp.~1058--1076, June 2010.

\bibitem{alkhateeb2014channel}
A.~Alkhateeb, O.~El~Ayach, G.~Leus, and R.~W. Heath, ``Channel estimation and
  hybrid precoding for millimeter wave cellular systems,'' {\em IEEE Journal of
  Selected Topics in Signal Processing}, vol.~8, no.~5, pp.~831--846, 2014.

\bibitem{donoho2006compressed}
D.~L. Donoho, ``Compressed sensing,'' {\em IEEE Transactions on information
  theory}, vol.~52, no.~4, pp.~1289--1306, 2006.

\bibitem{tropp2004greed}
J.~A. Tropp, ``Greed is good: Algorithmic results for sparse approximation,''
  {\em IEEE Transactions on Information theory}, vol.~50, no.~10,
  pp.~2231--2242, 2004.

\bibitem{candes2006robust}
E.~J. Cand{\`e}s, J.~Romberg, and T.~Tao, ``Robust uncertainty principles:
  Exact signal reconstruction from highly incomplete frequency information,''
  {\em IEEE Transactions on information theory}, vol.~52, no.~2, pp.~489--509,
  2006.

\bibitem{boyd2004convex}
S.~Boyd and L.~Vandenberghe, {\em Convex optimization}.
\newblock Cambridge university press, 2004.

\bibitem{8109448}
A.~Jurenoks and L.~Novickis, ``Analysis of wireless sensor network structure
  and life time affecting factors,'' in {\em Communication and Information
  Technologies (KIT)}, pp. 1-6, Oct 2017.

\bibitem{926982}
W.~B. Heinzelman, A.~Chandrakasan, and H.~Balakrishnan, ``Energy-efficient
  communication protocol for wireless microsensor networks,'' in {\em
  Proceedings of the 33rd Annual Hawaii International Conference on System
  Sciences}, vol. 2, Jan 2000.

\bibitem{925197}
A.~Manjeshwar and D.~P. Agrawal, ``Teen: a routing protocol for enhanced
  efficiency in wireless sensor networks,'' in {\em Proceedings 15th
  International Parallel and Distributed Processing Symposium (IPDPS)}, pp.
  2009-2015, April 2001.

\bibitem{smaragdakis2004sep}
G.~Smaragdakis, I.~Matta, and A.~Bestavros, ``Sep: A stable election protocol
  for clustered heterogeneous wireless sensor networks,'' tech. rep., Boston
  University Computer Science Department, 2004.

\bibitem{QING20062230}
L.~Qing, Q.~Zhu, and M.~Wang, ``Design of a distributed energy-efficient
  clustering algorithm for heterogeneous wireless sensor networks,'' {\em
  Computer Communications}, vol.~29, no.~12, pp.~2230 -- 2237, 2006.

\bibitem{1696382}
A.~Khadivi and M.~Shiva, ``Ftpasc: A fault tolerant power aware protocol with
  static clustering for wireless sensor networks,'' in {\em IEEE International
  Conference on Wireless and Mobile Computing, Networking and Communications},
  pp. 397-401, June 2006.

\bibitem{SWQ28921546.OW12N}
I.~Azam, A.~Majid, I.~Ahmad, U.~Shakeel, H.~Maqsood, Z.~A. Khan, U.~Qasim, and
  N.~Javaid, ``Seec: Sparsity-aware energy efficient clustering protocol for
  underwater wireless sensor networks,'' in {\em IEEE 30th International
  Conference on Advanced Information Networking and Applications (AINA)}, pp.
  352-361, March 2016.

\bibitem{5136647}
F.~Bajaber and I.~Awan, ``Centralized dynamic clustering for wireless sensor
  network,'' in {\em International Conference on Advanced Information
  Networking and Applications (AINA) Workshops}, pp. 193-198, May 2009.

\bibitem{4809826}
G.~S. Tomar and S.~Verma, ``Dynamic multi-level hierarchal clustering approach
  for wireless sensor networks,'' in {\em 11th International Conference on
  Computer Modelling and Simulation}, pp. 563-567, March 2009.

\bibitem{7822956}
M.~K. Naeem, M.~Patwary, and M.~Abdel-Maguid, ``Universal and dynamic
  clustering scheme for energy constrained cooperative wireless sensor
  networks,'' {\em IEEE Access}, vol.~5, pp.~12318--12337, 2017.

\bibitem{7365420}
D.~Jia, H.~Zhu, S.~Zou, and P.~Hu, ``Dynamic cluster head selection method for
  wireless sensor network,'' {\em IEEE Sensors Journal}, vol.~16, pp.
  2746-2754, April 2016.

\bibitem{S2RT4387429.OW12N}
A.~Ahmad, K.~Latif, N.~Javaidl, Z.~A. Khan, and U.~Qasim, ``Density controlled
  divide-and-rule scheme for energy efficient routing in wireless sensor
  networks,'' in {\em 26th IEEE Canadian Conference on Electrical and Computer
  Engineering (CCECE)}, pp. 1-4, May 2013.

\bibitem{6531736}
A.~S.~K. Mammu, A.~Sharma, U.~Hernandez-Jayo, and N.~Sainz, ``A novel
  cluster-based energy efficient routing in wireless sensor networks,'' in {\em
  IEEE 27th International Conference on Advanced Information Networking and
  Applications (AINA)}, pp. 41-47, March 2013.

\end{thebibliography}

\end{document}